\documentclass[lettersize,onecolumn]{IEEEtran}
\usepackage{amsmath,amsfonts}
\usepackage{algorithmic}
\usepackage{algorithm}
\usepackage{array}
\usepackage[caption=false,font=normalsize,labelfont=sf,textfont=sf]{subfig}
\usepackage{textcomp}
\usepackage{stfloats}
\usepackage{url}
\usepackage{verbatim}
\usepackage{graphicx}
\usepackage{amsthm,amsmath,amsfonts,amssymb,breqn,mathtools}
\usepackage{natbib, hyperref}
\usepackage[dvipsnames]{xcolor}
\usepackage{comment}
\usepackage{calc}  
\usepackage{color}

\theoremstyle{plain}
\newtheorem{theorem}{Theorem}[section]
\newtheorem{lemma}{Lemma}[section]
\newtheorem{proposition}{Proposition}[section]
\newtheorem{corollary}{Corollary}[section]

\theoremstyle{remark}

\newcommand{\Pro}{\mathrm{P}}
\newcommand{\Qro}{\mathrm{Q}}
\newcommand{\Exp}{\mathrm{E}}

\newcommand{\cL}{\mathcal{L}}
\newcommand{\cF}{\mathcal{F}}

\newcommand{\cA}{\mathcal{A}}

\DeclareMathOperator*{\argmax}{arg\,max\,}

\newcommand{\btheta}{\boldsymbol{\theta}}
\newcommand{\bTheta}{\boldsymbol{\Theta}}
\newcommand{\bH}{\boldsymbol{H}}
\newcommand{\bA}{\boldsymbol{A}}
\newcommand{\bPro}{\boldsymbol{\Pro}}
\newcommand{\bExp}{\boldsymbol{\Exp}}
\newcommand{\bD}{\boldsymbol{D}}

\newcommand{\balpha}{\boldsymbol{\alpha}}
\newcommand{\ba}{\boldsymbol{a}}
\newcommand{\bell}{\boldsymbol{\ell}}

\newcommand{\bB}{\boldsymbol{B}}

\newcommand{\bI}{\boldsymbol{I}}

\newcommand{\midbig}[1]{
  \mathopen{}\mathclose\bgroup
  \hbox{\scalebox{1.2}{$\displaystyle#1$}}
  \egroup
}
\newcommand{\ano}{l} 
\newcommand{\Alt}{\operatorname{Alt}}

\newcommand{\bcF}{\boldsymbol{\cF}}

\begin{document}

\title{Sequential multiple testing with multiple hypotheses and prior information on the hypothesis configuration}

\author{Yiming Xing,~\IEEEmembership{Member,~IEEE}
\thanks{
Yiming Xing is with School of Mathematical Sciences, Tongji University, Shanghai, China.
}
\thanks{A special case of this work was presented at 2026 IEEE International Conference on Acoustics, Speech, and Signal Processing \citep{ICASSP2026}.}
\thanks{Manuscript received ...; revised ...}}


\markboth{Journal of \LaTeX\ Class Files,~Vol.~1, No.~2, December~2023}%
{Shell \MakeLowercase{\textit{et al.}}: A Sample Article Using IEEEtran.cls for IEEE Journals}

\IEEEpubid{0000--0000~\copyright~2023 IEEE}

\maketitle

\begin{abstract}
    In this work, we study the problem of testing the marginal distributions of multiple independent, sequentially observed data streams, 
    where for each stream there are multiple candidate hypotheses to select from, in the presence of prior information on the unknown hypothesis configuration.
    The goal is to understand the benefit of such information and to design a sequential testing procedure that effectively leverages it. 
    We start with arbitrary prior information and specialize to concrete examples, including known number or known lower bound on the number of streams following each hypothesis, and the presence of exclusive hypotheses.
    The designed procedure is three-fold:
    (i) reliable, i.e., controlling all types of familywise error probabilities below arbitrary user-specified levels,
    (ii) computationally efficient, i.e., focusing on minimal sets of alternative hypothesis configurations in making decisions,
    and
    (iii) asymptotically optimal, i.e., achieving the minimum expected sample size among all reliable procedures asymptotically as the error levels go to zero. 
    Numerical studies are presented for illustration.
\end{abstract}


\section{Introduction}
Suppose there are multiple independent and sequentially observed data streams. We aim to simultaneously test the marginal distribution for each of them, 
with the goal of controlling the probabilities of making wrong testing decisions below desired levels and minimizing the number of observations required on average.
Such a problem arises in various real-world scenarios, such as multi-channel signal detection \cite{spectrumsensing_2012, Fellouris_noniid}, 
multi-sensor anomaly identification \cite{Ref4AnomalyDet_2009, Kobi_2015_active}, 
multi-endpoint clinical trials \cite{berner2007clinical, Bartroff_Book_Clinicaltrials}, etc.


Most of existing works about this problem focused on the case where there are two candidate hypotheses for each data stream, 
e.g., \cite{Malloy_Nowak_2014, Xiang_stable2023}, \cite{Bartroff_2014_FWER, Bartroff_2018_GFWER, Bartroff_2021_EquivErrorMetrics}, \cite{Kobi_2015_active, Kobi_2018_heterogeneous, Kobi_2020_composite, Kobi_2023_hiera}, and \cite{Song_prior, Song_AoS, Aris_IEEE, Aris_TIT2025, PaperII, PaperIII, chaudhuri2024joint, ITW2024}.
However, many application scenarios feature multiple candidate hypotheses, e.g., when there is an intermediate state between the normal and the abnormal states, or when we need to categorize signals into several groups. 

The first contribution of this work is to extend the problem of sequential multiple testing to multiple candidate hypotheses in each data stream. 
To the best of the authors' knowledge, this general formulation has previously been considered only by \cite{ITW2024}, in a special decentralized setup where streams cannot share information with each other. 
Specifically, in this work, we investigate the problem of simultaneously solving multiple hypothesis testing problems, each involving multiple hypotheses, for the marginal distributions of multiple independent data streams, where
there is a centralized decision maker that collects information from all streams to determine when to stop and how to make the decisions. 


Another perspective on the relationship between this work and existing literature is that this work extends from solving ``one" hypothesis testing problem with multiple candidate hypotheses for the distribution of ``one" data stream to simultaneously solving ``multiple" such hypothesis testing problems for the distributions of ``multiple" data streams.
The former has been studied, e.g., by \cite{Chernoff1959, Multihyp_Tart_1999, Multihyp_Lai_2000, ControlledSensing, ControlledSensing_NonUniformCost, ControlledSensing_Composite, ISIT2024, ISIT2025}.
However, extending it to multiple data streams and multiple hypothesis testing problems presents many open questions, e.g., (i) how to utilize prior information on the unknown hypothesis configuration \cite{Kobi_2015_active, Song_prior}, 
which is exactly what we aim to address in this work, 
(ii) how to define and control generalized error metrics that tolerate a small number of errors \cite{Song_AoS, PaperII}, 
(iii) what if there is a sampling constraint that allows for observing only a subset of streams at every time instant \cite{Aris_IEEE, Aris_TIT2025}, and (iv) how to utilize / adapt to the dependence \cite{chaudhuri2024joint} or hierarchical \cite{Kobi_2023_hiera} structure among streams, etc.

The second contribution of this work is the solution of the first open question, i.e., how to take advantage of prior information on the true, unknown hypothesis configuration.
We mean by hypothesis configuration the unknown fact of which hypothesis each stream follows, which is exactly what we aim to infer based on collecting data and conducting a test.
The incorporation of prior information is motivated by numerous application contexts. 
The two most common forms of prior information considered in the literature, both restricted to the case of two candidate hypotheses for each stream, are 
known number of streams following each hypothesis, e.g., \cite{Malloy_Nowak_2014, Kobi_2015_active, Kobi_2018_heterogeneous, Kobi_2020_composite, Kobi_2023_hiera}, 
and known lower bound on the number of streams following each hypothesis, e.g., \cite{Song_prior, Bartroff_2021_EquivErrorMetrics, PaperII, PaperIII, Aris_IEEE, chaudhuri2024joint}.
The former arises, e.g., when, after certain detection procedure, we know that there exists exactly one signal and we need to correctly identify it, and the latter arises when we know signals exist, but their number is unknown. 
In this work, we will first focus on arbitrary prior information, and then specialize our general results to
four concrete examples: 
(i) no prior information, (ii) known number of streams following each hypothesis, (iii) known lower bound on the number of streams following each hypothesis, 
and (iv) the existence of exclusive hypotheses such that if streams following one hypothesis exist then streams following the other cannot.
Note that all of these are studied with multiple hypotheses for each stream.

To see how multiple hypotheses substantially complicate the problem, let us take the case of known numbers as an example. 
Specifically, when there are two hypotheses in each stream and the number of streams following each hypothesis is known, to make reliable decisions it suffices to sample until the streams form two groups whose sizes are consistent with the prior information and the gap of the log-likelihood ratios between these two groups is sufficiently large. 
This idea has been fully exploited in the works cited above. 
However, with three hypotheses, apart from requiring that the streams form three groups of correct sizes and that the pairwise gaps between every two groups are sufficiently large, we also need to require that two ``cyclic gaps" among all three groups are sufficiently large as well.
Briefly, the reason is that pairwise gaps avoid making pairwise errors of type $1\leftrightarrow 2$, $2\leftrightarrow 3$, and $3\leftrightarrow 1$, while cyclic gaps avoid making cyclic errors of type $1\rightarrow 2\rightarrow 3\rightarrow 1$ and $1\rightarrow 3\rightarrow 2\rightarrow 1$.
More hypotheses lead to more cycles and longer cycle lengths.
This will become clear after the proposed testing procedure in Section \ref{sec: methodology} and the examples in Section \ref{sec: examples}.

The rest of this work is organized as follows:
In Section \ref{sec: problem formulation} we formulate the problem of sequential, multistream and multihypothesis testing with prior information. 
In Section \ref{sec: universal lower bound} we establish a universal lower bound on the optimal expected sample size, which motivates the design of our procedure.
In Section \ref{sec: methodology} we introduce the proposed procedure and analyze its properties. 
In Section \ref{sec: examples} we specialize to four concrete examples. 
In Section \ref{sec: numerical studies} we present numerical studies.
In Section \ref{sec: extension to composite hypotheses} we discuss the extension to parametric composite hypotheses. 
In Section \ref{sec: conclusion} we conclude and pose some future research directions. 
Most proofs are deferred to the Appendix. 

\section{Problem formulation} \label{sec: problem formulation}
Let $X^k:=\{X^k(n):n\geq 1\}$, $k\in[K]:=\{1,\ldots,K\}$ be $K\geq 2$ independent streams of i.i.d. data. 
For every $k\in[K]$, denote by $f^k$ the density of $X^k(1)$ with respect to some $\sigma$-finite measure $\nu^k$ and consider the following hypothesis testing problem about it:
\begin{equation} \label{testing problem}
    f^k=f^k_i \text{ for } i\in[M],
\end{equation}
where $M\geq 2$ and $\{f^k_i:i\in[M]\}$ are distinct. 
Our only assumption throughout this work is that, the Kullback-Leibler divergence between $f^k_i$ and $f^k_j$ for any $k\in[K]$ and $i\neq j\in[M]$ are positive and finite, i.e., 
\begin{equation} \label{only assumption}
    I^k_{i,j} := \int f^k_i \log\frac{f^k_i}{f^k_j} d\nu^k\in (0,\infty).    
\end{equation}

With an abuse of notation, we use
\begin{equation*}
    \bH = (H^1,\ldots,H^K) = (H_1,\ldots,H_M), 
\end{equation*}
to denote both the true hypothesis of every stream and the true subset of streams following each hypothesis, i.e., for every $k\in[K]$ and $i\in[M]$, $H^k=i$ if and only if $k\in H_i$ if and only if $f^k=f^k_i$. 
We refer to $\bH$ as the \emph{hypothesis configuration}, or simply the \emph{configuration}.
We use $[M]^K$ to denote all possible
configurations, i.e.,
\begin{equation*}
\begin{aligned}
    [M]^K & := \{(H^1,\ldots,H^K):H^k\in[M] \text{ for all } k\in[K]\} \\
    & := \{(H_1,\ldots,H_M): H_i\cap H_j=\emptyset \text{ for all } i\neq j\in[M] \text{ and } \cup_{i\in[M]} H_i=[K] \}.
\end{aligned}
\end{equation*}
For any $k\in[K]$ and $i\in[M]$, we denote by $\Pro^k_i$ the distribution of $X^k$ when $f^k=f^k_i$ and, for any $\bH\in[M]^K$, we denote by $\bPro_\bH=\prod_{k\in[K]} \Pro^k_{H^k}$ the distribution of all streams when $f^k=f^k_{H^k}$ for every $k\in[K]$, where the product holds because of the assumption of independence across streams.
We use $\Exp^k_i$ and $\bExp_\bH$ to denote the corresponding expectations. 

Suppose that data are observed sequentially in time, i.e., at every time $n\geq 1$, we observe $X(n):=\{X^k(n):k\in[K]\}$. 
Denote by $\bcF:=\{\bcF(n):n\geq 1\}$ the filtration induced by the data, i.e., $\bcF(0)$ is the trivial $\sigma$-algebra and, for every $n\geq 1$, $\bcF(n):=\sigma(\bcF(n-1), X(n))$ is the $\sigma$-algebra containing all information available from the data up to time $n$. 

A solution to this sequential, multistream, and multihypothesis testing in \eqref{testing problem}, referred to as a \emph{testing procedure}, or simply a \emph{procedure}, should consist of a random time $T$ that indicates when to stop sampling, and an $[M]^K$-valued random element $\bD=(D^1,\ldots,D^K)=(D_1,\ldots,D_M)$ that indicates how to make the decisions.
The interpretation is that, after taking $n$ data from every stream, hypothesis $D^k$ is selected for stream $k$ where $k\in[K]$, i.e., hypothesis $i$ is selected for streams in $D_i$ where $i\in[M]$.
Since all actions can utilize only available information, 
$T$ is required to be an $\bcF$-stopping time and $\bD$ be $\bcF(T)$-measurable, i.e., $\{T=n,\,\bD=\bA\}\in\bcF(n)$ for all $n\geq 1$ and $\bA\in[M]^K$.
We denote by $(T,\bD)$ a procedure and by $\Delta$ the family of all procedures. 

When the true configuration is $\bH$ and procedure $(T,\bD)$ is used, for every $i\neq j\in[M]$, $H_i\cap D_j$ represents those streams that follow hypothesis $i$ but are misidentified as following hypothesis $j$, which we refer to as \emph{type-(i,j) errors}. 
We aim to control the probability of making type-$(i,j)$ errors, i.e., $\bPro_\bH(H_i\cap D_j\neq\emptyset)$, for all possible configuration $\bH$ and all $i\neq j\in[M]$. 
This should be the most refined error metric, 
and can be reduced to controlling the probability of making \emph{type-$(i,)$ errors}, i.e., $\bPro_\bH(H_i\cap(\cup_{j\in[M]\backslash\{i\}} D_j) \neq \emptyset) = \bPro_\bH(H_i\backslash D_i\neq\emptyset)$, for all $i\in[M]$, and to controlling the probability of making \emph{type-$(,j)$ errors}, i.e., $\bPro_\bH((\cup_{i\in[M]\backslash\{j\}} H_i) \cap D_j \neq \emptyset) = \bPro_\bH(D_j\backslash H_j\neq\emptyset)$, for all $j\in[M]$.

To incorporate prior information about the true, unknown configuration $\bH$, we assume that there is a subset of $[M]^K$, denoted as $\cA$, such that it is a priori known that $\bH\in\cA$.
In this work, we first investigate the general version of $\cA$, and then specialize to the following four examples:
\begin{itemize}
    \item [(I)] No priori information, i.e., $\cA=[M]^K$. 
    \item [(II)] Known number of streams following each hypothesis, i.e., $\cA=\cA^{\text{exact}}_{K_1,\ldots,K_M}$, where
    \begin{equation} \label{def, known exact}
        \cA^{\text{exact}}_{K_1,\ldots,K_M} := \{ \bA\in[M]^K: |A_i|=K_i \,\,\forall\, i\in[M] \},
    \end{equation}
    for some $K_1,\ldots,K_M\geq 1$ that $K_1+\cdots+K_M=K$.
    Note that the case where $K_i=0$ for some $i\in[M]$ is excluded, because if so we may simply remove that hypothesis in all streams. 
    \item [(III)] Known lower bounds on the number of streams following each hypothesis, i.e., $\cA=\cA^{\text{lower}}_{L_1,\ldots,L_M}$, where 
    \begin{equation} \label{def, known lower}
        \cA^{\text{lower}}_{L_1,\ldots,L_M} := \{ \bA\in[M]^K: |A_i|\geq L_i \,\,\forall\, i\in[M] \},
    \end{equation}
    for some $L_1,\ldots,L_M\geq 0$ that $L_1+\cdots+L_M\leq K$. Note that when $L_1=\cdots=L_M=0$, this coincides with the first example, and when $L_1+\cdots+L_M=K$, this coincides with the second. 
    \item [(IV)] There exist two hypotheses that exclude each other, in the sense that if there exist streams following one hypothesis then there cannot exist streams following the other, i.e., $\cA=\cA^{\text{exclu}}_{i,j}$, where 
    \begin{equation} \label{def, exclu hyps}
        \cA^{\text{exclu}}_{i,j} := \{ \bA\in[M]^K: A_i=\emptyset \text{ or } A_j=\emptyset \}, 
    \end{equation}
    for some $i\neq j\in[M]$.
\end{itemize}
Of course, more interesting and practical cases are ready to be explored, which can be done analogously based on our general results.


We use $\balpha=(\alpha_{i,j})_{i\neq j\in[M]}\in(0,1)^{M(M-1)}$ to denote the tolerance levels of error probabilities. 
For any prior information $\cA$ and error levels $\balpha$, we use $\Delta(\balpha,\cA)$ to denote the subfamily of procedures that terminate almost surely in finite time and control the error probabilities of all types below the corresponding levels in $\balpha$ simultaneously for all configurations consistent with prior information $\cA$, i.e., 
\begin{equation*}
\begin{aligned}
    \Delta(\balpha,\cA) := \{  
    (T,\bD)\in\Delta: \, 
    & \bPro_\bH(T<\infty) = 1 \text{ and} \\
    & \bPro_\bH(H_i\cap D_j\neq\emptyset)\leq\alpha_{i,j} \text{ for all } \bH\in\cA \text{ and } i\neq j\in[M]
    \}.
\end{aligned}
\end{equation*}
Under these reliability constraints, our goal is to minimize the expected sample size, i.e., 
\begin{equation*}
    \cL_\bH(\balpha,\cA) := \inf\{ \bExp_\bH[T]: (T,\bD)\in\Delta(\balpha,\cA) \},
\end{equation*}
simultaneously for all $\bH\in\cA$, to a first-order asymptotic approximation as $\alpha_\max:=\max_{i\neq j\in[M]} \alpha_{i,j} \to 0$.

\section{Universal lower bound} \label{sec: universal lower bound}
In this section, we establish a universal lower bound on $\cL_\bH(\balpha,\cA)$, which will inspire the design of the proposed procedure.
To do this, we first introduce some notations.

For any stream $k\in[K]$, hypothesis $i\in[M]$ and time $n\geq 1$, denote by $\ell^k_i(n)$ the log-likelihood in stream $k$ for hypothesis $i$ based on its first $n$ data, i.e., 
\begin{equation*}
    \ell^k_i(n) := \sum_{t=1}^n \log f^k_i(X^k(t)).
\end{equation*}
For any configuration $\bH\in[M]^K$, denote by $\bell_\bH(n)$ the log-likelihood for configuration $\bH$, which, due to the assumption of independence across streams, can be written as the following summation:
\begin{equation*}
    \bell_\bH(n) := \sum_{k\in[K]} \ell^k_{H^k}(n).
\end{equation*}
Note that, according to condition \eqref{only assumption}, we have 
\begin{equation*}
    \Exp^k_i[\ell^k_i(1)-\ell^k_j(1)] = I^k_{i,j}
\end{equation*}
for any $k\in[K]$ and $i\neq j\in[M]$, and 
\begin{equation*}
    \bExp_\bH[\bell_\bH(1)-\bell_\bA(1)] = \sum_{k\in\bH\triangle\bA} I^k_{H^k,A^k} := \bI_\bH(\bA)
\end{equation*}
for any $\bH,\bA\in[M]^K$, where $\bH\triangle\bA:=\{k\in[K]: H^k\neq A^k\}$ represents those streams where the two configurations differ. 

To simplify the form of the lower bound, we denote, for any prior information $\cA$, configuration $\bH\in\cA$, and $i\neq j\in[M]$,
\begin{equation*}
    \Alt_{i,j}(\bH,\cA) := \{ \bA\in\cA: H_i\cap A_j\neq\emptyset \},
\end{equation*}
which represents all alternative configurations that are consistent with the prior information and make at least one type-$(i,j)$ error, possibly more in number and in type, if the truth is $\bH$.
Soon we will realize that it is more precise to understand $\Alt_{i,j}(\bH,\cA)$ as all alternative configurations that are consistent with the prior information and \emph{relative to which $\bH$ makes at least one type-$(j,i)$ error}. 
This follows from the fact that $\bA\in\Alt_{i,j}(\bH,\cA)$ is equivalent to $\bH\in\Alt_{j,i}(\bH,\cA)$.
Note that $\{\Alt_{i,j}(\bH,\cA):i\neq j\in[M]\}$ are not necessarily disjoint and their union is $\cA\backslash\{\bH\}$.
Further, we denote
\begin{equation} \label{bIij(bHcA)}
    \bI_{i,j}(\bH,\cA) := 
    \min\limits_{\bA\in\Alt_{i,j}(\bH,\cA)} \bI_\bH(\bA),
\end{equation}
which can be understood as the minimum \emph{distance} between $\bH$ and $\Alt_{i,j}(\bH,\cA)$.
Moreover, we denote
\begin{equation*}
    \varphi(x,y) := x\log\frac{x}{1-y} + (1-x)\log\frac{1-x}{y} \text{ for } x,y\in(0,1/2),
\end{equation*}
which represents the Kullback-Leibler divergence between two Bernoulli distributions with success probability $x$ and $1-y$ respectively.
Note that it satisfies $\varphi(x,y)\sim |\log y|$ as $x,y\to 0$.

\begin{theorem} \label{theorem, LB}
    For any prior information $\cA$, hypothesis configuration $\bH\in\cA$, and error probabilities $\balpha$, we have
    \begin{equation*} 
        \cL_\bH(\balpha,\cA) \geq \max_{i\neq j\in[M]} \frac{\varphi(\alpha_{\operatorname{sum}},\alpha_{j,i})}{\bI_{i,j}(\bH,\cA)},
    \end{equation*}
    where $\alpha_{\operatorname{sum}}:=\sum_{i\neq j\in[M]}\alpha_{i,j}$.
    Thus, as $\alpha_\max\to 0$, we have
    \begin{equation*}
        \cL_\bH(\balpha,\cA) \gtrsim \max_{i\neq j\in[M]} \frac{|\log\alpha_{j,i}|}{\bI_{i,j}(\bH,\cA)}.
    \end{equation*}
\end{theorem}
\begin{proof}
    See Appendix \ref{appendix, proofs}.
\end{proof}

Here is an interpretation of this lower bound: 
Ideally, we would decide on the correct configuration $\bH$ after collecting enough evidence that ensures safety against all types of errors. 
For every $i\neq j\in[M]$, $|\log\alpha_{j,i}|$ represents the amount of evidence needed against type-$(j,i)$ errors, and $\bI_{i,j}(\bH,\cA)$ represents the minimum accumulation rate of such evidence, i.e., evidence against those configurations relative to which deciding on $\bH$ results in type-$(j,i)$ errors, so the ratio $|\log\alpha_{j,i}|/\bI_{i,j}(\bH,\cA)$ represents the average time required.
Since we need to ensure safety against all types of errors, the maximum over $i\neq j\in[M]$ appears.
In practice, the correct configuration is unknown and has to be replaced by its maximum likelihood estimate, which will be elaborated in the next section.

\section{Methodology} \label{sec: methodology}
In this section, we introduce the procedures.
We start from a na\"ive procedure with high computational complexity, based on which we develop a substantially more efficient one.

Before starting, we introduce the following notations:
For any stream $k\in[K]$ and time $n\geq 1$, we denote by $\hat H^k(n)$ the \emph{maximum likelihood hypothesis (MLH)} for stream $k$ at time $n$, i.e., 
\begin{equation*}
    \hat H^k(n) := \argmax_{i\in[M]} \ell^k_i(n).
\end{equation*}
Besides, we denote by $\hat\bH(n)$ the \emph{maximum likelihood configuration (MLC)} at time $n$, i.e., 
\begin{equation*}
    \hat\bH(n) := (\hat H^1(n),\ldots,\hat H^K(n)) = (\hat H_1(n),\ldots,\hat H_M(n)),
\end{equation*}
where, for every $i\in[M]$, $\hat H_i(n):=\{k\in[K]:\hat H^k(n)=i\}$.

\subsection{A na\"ive procedure} 

Based on the intuition of sequential testing and the categorization of $\cA\backslash\{\bH\}$ into $\{\Alt_{i,j}(\bH,\cA):i\neq j\in[M]\}$, it is natural to come up with the following procedure:
\begin{equation} \label{a naive procedure}
\begin{aligned}
    \hat T & := \inf_{n\geq 1} \Big\{
    \hat\bH(n) \in \cA \text{ and, for all } i\neq j\in[M], \, \bell_{\hat\bH(n)}(n) - \max_{\bA\in\Alt_{i,j}(\hat\bH(n),\cA)} \bell_\bA(n) \geq a_{j,i}
    \Big\},  \\
    \hat\bD & := \hat\bH(\hat T), 
\end{aligned}
\end{equation}
where $\ba=(a_{i,j})_{i\neq j\in[M]}\in(0,\infty)^{M(M-1)}$ are thresholds to be determined. 
Specifically, this procedure samples until the MLC is consistent with the prior information and the evidence in favor of the MLC and against all alternative configurations in $\Alt_{i,j}(\hat\bH(n),\cA)$ exceeds threshold $a_{j,i}$ simultaneously for all $i\neq j\in[M]$, at which time the MLC is selected as the final decision. 
Here we note again that it is more precise to interpret $\Alt_{i,j}(\hat\bH(n),\cA)$ as all configurations relative to which type-$(j,i)$ errors would be made if we decide on $\hat\bH(n)$, 
so to ensure evidence against them exceeding threshold $a_{j,i}$
is to ensure safety against type-$(j,i)$ errors.

However, this procedure compares the MLC $\hat\bH(n)$ with all alternative configurations in $\cA\backslash\{\hat\bH(n)\}$, whose size is $|\cA|-1$, which can be as large as the order of $M^K$.
Besides, these many comparisons seem unnecessary. 
E.g., when there is no prior information, the procedure in \eqref{a naive procedure} requires comparing $\hat\bH(n)$ with all configurations that differ with $\hat\bH(n)$ in one stream, in two streams, etc.
However, based on existing studies in the case of two candidate hypotheses for each stream (see, e.g., \cite{Song_prior}), comparing $\hat\bH(n)$ with the most adjacent configurations, i.e., those that differ with $\hat\bH(n)$ in exactly one stream, should suffice. 

Essentially, our goal then is, for every $i\neq j\in[M]$, to find a subset of $\Alt_{i,j}(\hat\bH(n),\cA)$ such that comparing $\hat\bH(n)$ with all configurations in this subset is equivalent to comparing $\hat\bH(n)$ with all configurations in $\Alt_{i,j}(\hat\bH(n),\cA)$.
If this subset is much simpler than the original one, the computational complexity can be greatly reduced. 
This is exactly the idea we pursue in designing the proposed procedure, which will be introduced in the next subsection.

\subsection{The proposed procedure}
Fix prior information $\cA$, hypothesis configuration $\bH\in\cA$, and $i\neq j\in[M]$.
To understand the composition of $\Alt_{i,j}(\bH,\cA)$, note that every $\bA\in\Alt_{i,j}(\bH,\cA)$ uniquely defines $\{\bH\triangle\bA,(A^k)_{k\in\bH\triangle\bA}\}$, i.e., the subset of streams that undergo changes relative to $\bH$ and the types of the changes.  
For any $\bA,\bB\in\Alt_{i,j}(\bH,\cA)$, we say $\bB$ undergoes strictly fewer changes than $\bA$,
if $\bH\triangle\bB\subsetneqq\bH\triangle\bA$ and $B^k=A^k$ for all $k\in\bH\triangle\bB$, 
i.e., the subset of streams that undergo changes is a strict subset and the types of the changes on this subset are exactly the same. 
Suppressing the dependence on $\cA,\bH$ and $i,j$ (which will be clear from context), we denote this relation as $\bB\prec\bA$.
We define $\widetilde\Alt_{i,j}(\bH,\cA)$ as those configurations that undergo the fewest changes, i.e., 
\begin{align} \label{tildeAlt}
    \widetilde\Alt_{i,j}(\bH,\cA) := \midbig\{ \bA\in\Alt_{i,j}(\bH,\cA):  \text{ there does not exist } \bB\in\Alt_{i,j}(\bH,\cA) \text{ such that } \bB\prec\bA \midbig\}.
\end{align}

First note that this definition is well-defined.
Indeed, if there exist two different sets, $\widetilde\Alt_{i,j}(\bH,\cA)$ and $\widetilde\Alt'_{i,j}(\bH,\cA)$, such that neither is a subset of the other and both satisfy the definition, let $\bA\in\widetilde\Alt_{i,j}(\bH,\cA)\backslash\widetilde\Alt'_{i,j}(\bH,\cA)$.
By the definition of $\widetilde\Alt'_{i,j}(\bH,\cA)$, there must exist $\bA'\in\widetilde\Alt'_{i,j}(\bH,\cA)$ such that $\bA'\prec\bA$.
Since $\bA'\in\widetilde\Alt_{i,j}(\bH,\cA)\cap\widetilde\Alt'_{i,j}(\bH,\cA)$ implies $\bA,\bA'\in\widetilde\Alt_{i,j}(\bH,\cA)$, which is prohibited by the definition of $\widetilde\Alt_{i,j}(\bH,\cA)$, it must be that $\bA'\in\widetilde\Alt'_{i,j}(\bH,\cA)\backslash\widetilde\Alt_{i,j}(\bH,\cA)$.
On the other hand, by the definition of $\widetilde\Alt_{i,j}(\bH,\cA)$, there must exist $\bB\in\widetilde\Alt_{i,j}(\bH,\cA)$ such that $\bB\prec\bA'$. 
It can be verified that $\prec$ satisfies transitivity, so we have $\bA,\bB\in\widetilde\Alt_{i,j}(\bH,\cA)$ but $\bB\prec\bA'\prec\bA$, which is a contradiction.

The first good property of $\widetilde\Alt_{i,j}(\bH,\cA)$ is that, the minimum in the definition of $\bI_{i,j}(\bH,\cA)$ in \eqref{bIij(bHcA)} is attained on $\widetilde\Alt_{i,j}(\bH,\cA)$, i.e., 
\begin{equation} \label{bIij(bH,widetildecA)}
\begin{aligned}
    \bI_{i,j}(\bH,\cA) := \min\limits_{\bA\in\Alt_{i,j}(\bH,\cA)} \bI_\bH(\bA) = \min\limits_{\bA\in\widetilde\Alt_{i,j}(\bH,\cA)} \bI_\bH(\bA).
\end{aligned}
\end{equation}
To see this, it suffices to note that, for any $\bA,\bB\in\Alt_{i,j}(\bH,\cA)$ such that $\bB\prec\bA$, it holds that 
\begin{equation*}
\begin{aligned}
    \bI_\bH(\bA) = \sum\limits_{k\in\bH\triangle\bA} I^k_{H^k,A^k} & = \sum\limits_{k\in(\bH\triangle\bA)\backslash(\bH\triangle\bB)} I^k_{H^k,A^k} + \sum\limits_{k\in\bH\triangle\bB} I^k_{H^k,B^k} \\
    & \geq \sum\limits_{k\in\bH\triangle\bB} I^k_{H^k,B^k} = \bI_\bH(\bB).
\end{aligned}
\end{equation*}

The second good property is that, the maximum in the definition of $\hat T$ in \eqref{a naive procedure} is attained on $\widetilde\Alt_{i,j}(\hat\bH(n),\cA)$,
which is stated in the following proposition. 

\begin{proposition}
    For any prior information $\cA$ and $i\neq j\in[M]$, 
    \begin{equation*}
        \bell_{\hat\bH(n)}(n) - \max_{\bA\in\Alt_{i,j}(\hat\bH(n),\cA)} \bell_\bA(n) = 
        \bell_{\hat\bH(n)}(n) - \max_{\bA\in\widetilde\Alt_{i,j}(\hat\bH(n),\cA)} \bell_\bA(n).
    \end{equation*}
\end{proposition}
\begin{proof}
    Fix arbitrary $\cA$ and $i\neq j\in[M]$.
    It suffices to show that, for any $\bA,\bB\in\Alt_{i,j}(\hat\bH(n),\cA)$ such that $\bB\prec\bA$, 
    it holds that $\bell_{\hat\bH(n)}(n) - \bell_\bA(n)\leq \bell_{\hat\bH(n)}(n) - \bell_\bB(n)$.
    Indeed, 
    \begin{equation*}
    \begin{aligned}
        & \, \bell_{\hat\bH(n)}(n) - \bell_\bA(n) = \sum_{k\in\hat\bH(n)\triangle\bA}\Big( \ell^k_{\hat H^k(n)}(n) - \ell^k_{A^k}(n) \Big) \\
        = & \, \sum_{k\in(\hat\bH(n)\triangle\bA)\backslash(\hat\bH(n)\triangle\bB)}\Big( \ell^k_{\hat H^k(n)}(n) - \ell^k_{A^k}(n) \Big) + \sum_{k\in\hat\bH(n)\triangle\bB}\Big( \ell^k_{\hat H^k(n)}(n) - \ell^k_{B^k}(n) \Big) \\
        \geq & \, \sum_{k\in\hat\bH(n)\triangle\bB}\Big( \ell^k_{\hat H^k(n)}(n) - \ell^k_{B^k}(n) \Big) = \bell_{\hat\bH(n)}(n) - \bell_\bB(n), 
    \end{aligned}
    \end{equation*}
    where in the second step we used $\hat\bH(n)\triangle\bB\subsetneqq\hat\bH(n)\triangle\bA$ and $B^k=A^k$ on $\hat\bH(n)\triangle\bB$, and in the third step we used $\ell^k_{\hat H^k(n)}(n) = \max_{i\in[M]} \ell^k_i(n)$.
\end{proof}

A direct consequence of this proposition is that the na\"ive procedure in \eqref{a naive procedure} is equivalent to the following one:
\begin{equation} \label{the proposed procedure}
\begin{aligned}
    \hat T & := \inf_{n\geq 1} \Big\{
    \hat\bH(n) \in \cA \text{ and, for all } i\neq j\in[M], \, \bell_{\hat\bH(n)}(n) - \max_{\bA\in\widetilde\Alt_{i,j}(\hat\bH(n),\cA)} \bell_\bA(n) \geq a_{j,i}
    \Big\},  \\
    \hat\bD & := \hat\bH(\hat T), 
\end{aligned}
\end{equation}
which is the procedure we recommend in this work.
From now on, when we refer to the proposed procedure or $(\hat T,\hat\bD)$, we always mean the one in the form of \eqref{the proposed procedure}.

Note that, although the two expressions in \eqref{a naive procedure} and \eqref{the proposed procedure} are equivalent,
they lead to substantially different computational complexities. 
To see this, 
consider the example of no prior information, i.e., $\cA=[M]^K$.
In this case, for any $\bH\in[M]^K$ and $i\neq j\in[M]$, we have
\begin{equation} \label{Alt, no prior}
    \Alt_{i,j}(\bH,[M]^K) = \{\bA\in[M]^K: H_i\cap A_j\neq\emptyset\}, 
\end{equation}
and 
\begin{equation} \label{tildeAlt, no prior}
    \widetilde\Alt_{i,j}(\bH,[M]^K) = \{\bA\in[M]^K: \bH\triangle\bA = \{k\} \subseteq H_i\cap A_j\}.
\end{equation}
The size of the former is equal to $M^K$ minus the number of ways of putting $K$ different balls into $M$ different boxes such that $|H_i|$ specific balls are not put into one specific box, i.e., $M^K-(M-1)^{|H_i|}M^{K-|H_i|}$, which can be as large as $M^K-(M-1)^K$ when $|H_i|=K$,
whereas the size of the latter is simply $|H_i|$, which is less than or equal to $K$.

Before working out the specific form of $\widetilde\Alt_{i,j}(\bH,\cA)$ in concrete examples,
we first prove the reliability and asymptotic optimality of this procedure in the general case, which is the topic of the next two subsections.

\subsection{Error control}
In this subsection, we establish upper bounds on the error probabilities of the proposed procedure in \eqref{the proposed procedure} given thresholds, which in turn yield a selection of the thresholds in order to control the error probabilities below desired levels. 
Before starting, for any prior information $\cA$ and $i\neq j\in[M]$, we denote
\begin{equation*}
\begin{aligned}
    b_{i,j}(\cA) & := \max_{\bH\in\cA} \, |\Alt_{i,j}(\bH,\cA)|, \\
    \tilde b_{i,j}(\cA) & := \max_{\bH\in\cA} \, |\widetilde\Alt_{i,j}(\bH,\cA)|,
\end{aligned}
\end{equation*}
which is the maximum number of comparisons required in order to avoid type-$(j,i)$ errors if the procedure in \eqref{a naive procedure} or \eqref{the proposed procedure} is used, respectively.

\begin{theorem} \label{theorem, error control}
    For any prior information $\cA$, hypothesis configuration $\bH\in\cA$, and thresholds $\ba$, we have 
    \begin{equation*}
        \bPro_\bH(\hat T<\infty) = 1,
    \end{equation*}
    and 
    \begin{equation*}
        \bPro_\bH(H_i\cap\hat D_j\neq\emptyset) \leq |\Alt_{i,j}(\bH,\cA)| e^{-a_{i,j}} \text{ for all } i\neq j\in[M].
    \end{equation*}
    Therefore, for any error levels $\balpha$, $(\hat T,\hat\bD)\in\Delta(\balpha,\cA)$ if we select
    \begin{equation} \label{selection of thresholds}
        a_{i,j} = |\log\alpha_{i,j}| + \log b_{i,j}(\cA) \text{ for all } i\neq j\in[M].
    \end{equation}
\end{theorem}
\begin{proof}
    See Appendix \ref{appendix, proofs}.
\end{proof}

Note that, although the implementation of the proposed procedure requires only $\tilde b_{i,j}(\cA)$ comparisons, the universal selection of thresholds in \eqref{selection of thresholds} still involves the constant $b_{i,j}(\cA)$.
Similar to other works on sequential multiple testing, this selection, while sufficient for ensuring asymptotic optimality, can be very conservative in practice.
Thus, it is ideal to use Monte-Carlo simulations offline to find sharper thresholds that approximately equate the actual error probabilities with the target levels.
Since the target levels are usually small, estimating the error probabilities falls into the regime of rare event simulation, where importance sampling is an efficient tool.
Please see \cite{Siegmumd_IS, Bucklew_Book}, \cite[Chapter 4]{Song_prior} and \cite{Song_2025Efficient} for references. 
This will be further discussed in the numerical studies of Section \ref{sec: numerical studies}.


\subsection{Asymptotic optimality}
In this subsection, we establish an asymptotic upper bound on the expected sample size of the proposed procedure as its thresholds go to infinity. This, combined with the asymptotic lower bound on the minimum expected sample size and the reliability of the proposed procedure, constitutes a complete asymptotic optimality theory. 

\begin{theorem} \label{theorem, AUB}
    For any prior information $\cA$ and hypothesis configuration $\bH\in\cA$, we have
    \begin{equation} \label{expression, AUB}
        \bExp_\bH[\hat T] \lesssim \max_{i\neq j\in[M]} \frac{a_{j,i}}{\bI_{i,j}(\bH,\cA)}
    \end{equation}
    as thresholds $a_\min:=\min_{i\neq j\in[M]} a_{i,j}\to\infty$.
\end{theorem}
\begin{proof}
    See Appendix \ref{appendix, proofs}.
\end{proof}

\begin{corollary} \label{corollary, AO}
    For any prior information $\cA$, if we select the thresholds $\ba$ so that $(\hat T,\hat\bD)\in\Delta(\balpha,\cA)$ for all error levels $\balpha$ and $a_{i,j}\sim |\log\alpha_{i,j}|$ for all $i\neq j\in[M]$ as $\alpha_\max\to 0$, e.g., as in \eqref{selection of thresholds}, then for any $\bH\in\cA$, we have 
    \begin{equation*}
        \bExp_\bH[\hat T] \sim \max_{i\neq j\in[M]} \frac{|\log\alpha_{j,i}|}{\bI_{i,j}(\bH,\cA)} \sim
        \cL_\bH(\balpha,\cA)
    \end{equation*}
    as $\alpha_\max\to 0$.
\end{corollary}
\begin{proof}
    Combining Theorem \ref{theorem, LB}, \ref{theorem, error control} and \ref{theorem, AUB}.
\end{proof}

\section{Examples} \label{sec: examples}
In this section, we specialize the general results of the previous sections to concrete examples. 
We first introduce the following notation:
For any $i\neq j\in[M]$, denote
\begin{equation} \label{lambdai,j(n)}
    \lambda_{i,j}(n) := \min_{k\in\hat H_i(n)} \left\{ \ell^k_i(n) - \ell^k_j(n) \right\},
\end{equation}
which is the weakest evidence among those streams in favor of hypothesis $i$ against hypothesis $j$,
and, for any $\bH\in[M]^K$, denote
\begin{equation} \label{bIij(bH)}
    \bI_{i,j}(\bH) := \min_{k\in H_i} I^k_{i,j}.
\end{equation}

\subsection{No prior information}
The first example considers the case of no prior information, i.e., $\cA=[M]^K$.
The forms of $\Alt_{i,j}(\bH,[M]^K)$ and $\widetilde\Alt_{i,j}(\bH,[M]^K)$ for $i\neq j\in[M]$ have been derived in \eqref{Alt, no prior} and \eqref{tildeAlt, no prior}, 
and we have shown that $b_{i,j}([M]^K)=M^K-(M-1)^K\gg \tilde b_{i,j}([M]^K)=K$.

Note that the exact mathematical expression of $\widetilde\Alt_{i,j}(\bH,\cA)$, even in this simple example, is rather tedious.
Thus, we use $\{i\to j\}_\bH$ to mean that, relative to $\bH$, exactly one stream changes its hypothesis affiliation from $i$ to $j$ (while all other streams stay the same), 
and equate $\widetilde\Alt_{i,j}(\bH,[M]^K)=\{i\to j\}_\bH$.
Similarly, for any $2\leq m\leq M$ and distinct $i_1,\ldots,i_m\in[M]$, we use $\{i_1\to i_2\to\cdots\to i_m\}_\bH$ to mean that, relative to $\bH$, exactly one stream changes its hypothesis affiliation from $i_1$ to $i_2$, $\ldots$, and exactly one stream changes from $i_{m-1}$ to $i_m$, 
and use $\{i_1\to i_2\to\cdots\to i_m\to i_1\}_\bH$ to mean the same with the additional requirement that exactly one stream changes from $i_m$ to $i_1$, 
i.e., 
\begin{equation*}
\begin{aligned}
    & \, \{i_1\to i_2\to\cdots\to i_m\}_\bH \\
    := & \, \{ \bA\in[M]^K: \bH\triangle\bA=\{k_{i_1},\ldots,k_{i_{m-1}}\} \text{ where } k_{i_l}\in H_{i_l}\cap A_{i_{l+1}} \text{ for } 1\leq l\leq m-1 \}, \\
    & \, \{i_1\to i_2\to\cdots\to i_m\to i_1\}_\bH \\
    := & \, \{ \bA\in[M]^K: \bH\triangle\bA=\{k_{i_1},\ldots,k_{i_m}\} \text{ where } k_{i_l}\in H_{i_l}\cap A_{i_{l+1}} \text{ for } 1\leq l\leq m-1 \text{ and } k_{i_m}\in H_{i_m}\cap A_{i_1} \}.
\end{aligned}
\end{equation*}
Intuitively, the former means that exactly one stream in each of $H_{i_1},\ldots,H_{i_m}$ \emph{shift} their affiliations, 
and the latter means that exactly one stream in each of them \emph{rotate} their affiliations. 
In particular, we write $\{i\to j\to i\}_\bH$ as $\{i\leftrightarrow j\}_\bH$, which means that exactly one stream in $H_i$ and $H_j$ \emph{exchange} their affiliations. 


Meanwhile,
for an $\bA\in\{i\to j\}_\bH$ such that $\bH\triangle\bA=\{k\}\in H_i$, it is clear that 
$\bI_\bH(\bA)=I^k_{i,j}$, so 
\begin{equation*}
    \min_{\bA\in\{i\to j\}_\bH} \bI_\bH(\bA) = \bI_{i,j}(\bH),
\end{equation*}
where the latter was defined in \eqref{bIij(bH)},
and that
$\bell_\bH(n)-\bell_\bA(n) = \ell^k_i(n) - \ell^k_j(n)$, so
\begin{equation*}
    \bell_\bH(n) - \max_{\bA\in\{i\to j\}_\bH} \bell_\bA(n) = \min_{k\in H_i} \midbig\{ \ell^k_i(n) - \ell^k_j(n)\midbig\},
\end{equation*}
which is equal to $\lambda_{i,j}(n)$ in \eqref{lambdai,j(n)} if we replace $\bH$ by $\hat \bH(n)$.
Similarly, 
for an $\bA\in\{i_1\to i_2\to\cdots\to i_m\}_\bH$, where either $i_1,\ldots,i_m$ are distinct or $i_1,\ldots,i_{m-1}$ are distinct and $i_m=i_1$, such that $\bH\triangle\bA=\{k_{i_1},\ldots,k_{i_{m-1}}\}$ and 
$k_{i_l}\in H_{i_l}\cap A_{i_{l+1}}$ for $1\leq l\leq m-1$, 
we have $\bI_\bH(\bA)=\sum_{l=1}^{m-1} I^{k_{i_l}}_{i_l,i_{l+1}}$, so 
\begin{equation*}
    \min_{\bA\in\{i_1\to i_2\to\cdots\to i_m\}_\bH} \bI_\bH(\bA) = \sum_{l=1}^{m-1} \bI_{i_l,i_{l+1}}(\bH), 
\end{equation*}
and have
$\bell_\bH(n)-\bell_\bA(n) = \sum_{l=1}^{m-1} \midbig\{ \ell^{k_{i_l}}_{i_l}(n)-\ell^{k_{i_l}}_{i_{l+1}}(n) \midbig\}$, so 
\begin{equation*}
\begin{aligned}
    & \bell_\bH(n) - \max_{\bA\in\{i_1\to i_2\to\cdots\to i_m\}_\bH} \bell_\bA(n) 
    = \sum_{l=1}^{m-1} \min_{k_{i_l}\in H_{i_l}} \midbig\{ \ell^{k_{i_l}}_{i_l}(n) - \ell^{k_{i_l}}_{i_{l+1}}(n) \midbig\},
\end{aligned}
\end{equation*}
which is equal to $\sum_{l=1}^{m-1} \lambda_{i_l,i_{l+1}}(n)$ if we replace $\bH$ by $\hat\bH(n)$.
These notations and observations will greatly simplify our expressions in subsequent analysis. 

Returning to our current example with no prior information, it is clear that the generic expression of $\hat T$ in \eqref{the proposed procedure} can be specified as 
\begin{equation*} 
\begin{aligned}
    \hat T = & \inf_{n\geq 1}\midbig\{ \bell_{\hat\bH(n)}(n) - \max_{\bA\in\{i\to j\}_{\hat\bH(n)}} \bell_\bA(n) \geq a_{j,i} \text{ for all } i\neq j\in[M] \midbig\} \\
    = & \inf_{n\geq 1} \midbig\{ \lambda_{i,j}(n) \geq a_{j,i} \text{ for all } i\neq j\in[M] \midbig\}.
\end{aligned}
\end{equation*}
That is, we sample until the streams form $M$ groups, $\hat H_1(n),\ldots,\hat H_M(n)$, and for every group the minimum evidences in favor of its label and against all other labels exceed the corresponding thresholds.
In particular, when there is only one stream, i.e., $K=1$, this reduces to the multihypothesis sequential probability ratio test (see, e.g., \cite[Chapter 4]{Tartakovsky_Book}).
Besides, the constants in the asymptotic approximation to the optimal expected sample size in \eqref{bIij(bH,widetildecA)} are 
\begin{equation*}
    \bI_{i,j}(\bH,[M]^K) = \bI_{i,j}(\bH) \text{ for all } i\neq j\in[M].
\end{equation*}

\subsection{Known exact numbers}
The second example considers the case of known number of streams following each hypothesis, i.e., $\cA=\cA^{\text{exact}}_{K_1,\ldots,K_M}$ in \eqref{def, known exact}.
Throughout this subsection, we treat $K_1,\ldots,K_M$ as fixed and use $\cA^{\text{exact}}$ to denote the prior information. 
We also fix arbitrary $\bH\in\cA^{\text{exact}}$ and $i\neq j\in[M]$.

To study the composition of $\widetilde\Alt_{i,j}(\bH,\cA^{\text{exact}})$, let us first consider the case of two hypotheses in each stream, i.e., $M=2$.
It is clear that, in this case, when a type-$(1,2)$ error occurs relative to $\bH$, a type-$(2,1)$ error must occur simultaneously
so that the number of streams following each hypothesis remains consistent with the prior information, i.e., errors occur in \emph{pairs}.
Based on our notations, we have
\begin{equation*}
    \widetilde\Alt_{1,2}(\bH,\cA^{\text{exact}})=\widetilde\Alt_{2,1}(\bH,\cA^{\text{exact}})=\{1\leftrightarrow 2\}_\bH.
\end{equation*}
This observation has been exploited by numerous works that incorporate this kind of prior information into multiple testing with two hypotheses for each stream, e.g., \cite{Kobi_2015_AO, Song_prior, Bartroff_2021_EquivErrorMetrics}.

However, when there are multiple hypotheses in each stream, i.e., $M\geq 3$, errors occur not only in pairs, but also in \emph{cycles}.
E.g., if a stream in $H_i$ is misidentified as following $H_j$, a stream in $H_j$ is misidentified as following $H_\ano$, and a stream in $H_\ano$ is misidentified as following $H_i$, then a type-$(i,j)$ error occurs, the prior information is respected, but there is not a pair of streams in $H_i$ and $H_j$ that are exchanged, but a cycle of streams in $H_i, H_j$ and $H_\ano$ that are rotated.
Based on our notations, this form of cyclic errors is denoted as $\{i\to j\to l\to i\}_\bH$. 
This cycle can be of length up to $M$ and in arbitrary order, as long as it starts with $i\rightarrow j\rightarrow$ and ends with $\rightarrow i$.
Formally, $\widetilde\cA_{i,j}(\bH,\cA^{\text{exact}})$ consists of all configurations with exactly one such cycle relative to $\bH$, and we express it as
\begin{equation*}
\begin{aligned}
    & \widetilde\cA_{i,j}(\bH,\cA^{\text{exact}}) = \bigcup_{\substack{2\leq m\leq M \\ \text{distinct } i_1,\ldots,i_m\in[M] \\ \text{with } i_1=i,\,i_2=j }} \{i_1\to i_2\to\cdots\to i_m\to i_1\}_\bH.
\end{aligned}
\end{equation*}
Its size is
\begin{equation*}
\begin{aligned}
    & \, |\tilde\cA_{i,j}(\bH,\cA^{\text{exact}})| = 
    \sum_{m=2}^M \sum_{\substack{\text{distinct} \\ i_1,\ldots,i_m\in[M] \\ \text{with } i_1=i,\,i_2=j}} K_{i_1}\cdots K_{i_m}.
\end{aligned}
\end{equation*}
Since this number is the same for all $\bH\in\cA^{\text{exact}}$, it is also $\tilde b_{i,j}(\cA^{\text{exact}})$. 
To see its order of magnitude, note that 
\begin{equation*}
\begin{aligned}
    & \, \tilde b_{i,j}(\cA^{\text{exact}}) \leq K_1\cdots K_M \sum_{m=2}^M \frac{(M-2)!}{(M-m)!} \\
    \leq & \, K_1\cdots K_M \cdot e (M-2)! \leq \left(\frac{K}{M}\right)^M \cdot e(M-2)! \leq K^M.
\end{aligned}
\end{equation*}
Besides, the constant in the asymptotic approximation is
\begin{equation*}
\begin{aligned}
    \bI_{i,j}(\bH,\cA^{\text{exact}}) = \min_{\substack{2\leq m\leq M \\ \text{distinct } i_1,\ldots,i_m\in[M] \\ \text{with } i_1=i,\,i_2=j,\,i_{m+1}:=i}} \sum_{\ano=1}^m \bI_{i_\ano,i_{\ano+1}}(\bH).
\end{aligned}
\end{equation*}

Moreover, in this case the generic stopping time in \eqref{the proposed procedure} is equivalent to 
\begin{equation*}
\begin{aligned}
    \hat T  
    = & \inf_{n\geq 1} \Big\{ |\hat H_i(n)|=K_i \,\,\forall\, i\in[M], \text{ and} \\
    & \text{for all $2\leq m\leq M$ and distinct } i_1,\ldots,i_m\in[M] \text{ with } i_{m+1}:=i_1, \\
    & \sum_{\ano=1}^m \lambda_{i_\ano,i_{\ano+1}}(n) \geq \max_{1\leq\ano\leq m} a_{i_{\ano+1},i_\ano}
    \Big\}.
\end{aligned}
\end{equation*}
That is, we sample until the streams form $M$ groups of correct sizes, and the evidences against making all cyclic errors (of length $m$ and cycle $i_1\rightarrow i_2 \rightarrow \cdots \rightarrow i_m\rightarrow i_1$ for all $m$ and $i_1,\ldots,i_m$) exceed the corresponding thresholds ($a_{i_2,i_1}\vee\cdots\vee a_{i_m,i_{m-1}}\vee a_{i_1,i_m}$).
In particular, 
when $M=2$, we have
\begin{equation*}
\begin{aligned}
    \hat T = \midbig\{ |\hat H_i(n)| & =K_i \text{ for both } i\in[2], \text{ and} \\
    & \lambda_{1,2}(n) + \lambda_{2,1}(n) \geq a_{2,1}\vee a_{1,2}
    \midbig\},
\end{aligned}
\end{equation*}
and 
\begin{equation*}
    \bI_{1,2}(\bH,\cA^{\text{exact}}) = \bI_{2,1}(\bH,\cA^{\text{exact}}) = \bI_{1,2}(\bH) + \bI_{2,1}(\bH),
\end{equation*}
which recovers the ``gap rule" in \cite{Song_prior}.
When $M=3$, we have
\begin{equation} \label{hat T, known numbers and M=3}
\begin{aligned}
    \hat T = \inf_{n\geq 1} \midbig\{ |\hat H_i(n)|=K_i \text{ for all } & i\in[3], \text{ and } \\
    \lambda_{1,2}(n) + \lambda_{2,1}(n) & \geq a_{2,1} \vee a_{1,2}, \\
    \lambda_{1,3}(n) + \lambda_{3,1}(n) & \geq a_{3,1} \vee a_{1,3}, \\
    \lambda_{2,3}(n) + \lambda_{3,2}(n) & \geq a_{3,2} \vee a_{2,3}, \\
    \lambda_{1,2}(n) + \lambda_{2,3}(n) + \lambda_{3,1}(n) & \geq a_{2,1} \vee a_{3,2} \vee a_{1,3}, \\
    \lambda_{1,3}(n) + \lambda_{3,2}(n) + \lambda_{2,1}(n) & \geq a_{3,1} \vee a_{2,3} \vee a_{1,2}
    \midbig\}, 
\end{aligned}
\end{equation}
and, e.g., when $i=1$, $j=2$, 
\begin{equation} \label{bI, two known numbers}
\begin{aligned}
    \bI_{1,2}(\bH,\cA^{\text{exact}}) & = \min\midbig\{ \bI_{1,2}^{\text{len}=2}(\bH,\cA^{\text{exact}}), \bI_{1,2}^{\text{len}=3}(\bH,\cA^{\text{exact}}) \midbig\},
\end{aligned}
\end{equation}
where 
\begin{equation*}
\begin{aligned}
    \bI_{1,2}^{\text{len}=2}(\bH,\cA^{\text{exact}}) & = \bI_{1,2}(\bH) + \bI_{2,1}(\bH), \\
    \bI_{1,2}^{\text{len}=3}(\bH,\cA^{\text{exact}}) & = \bI_{1,2}(\bH) + \bI_{2,3}(\bH) + \bI_{3,1}(\bH).
\end{aligned}
\end{equation*}
Note that, because Kullback-Leibler divergences do not satisfy the triangle inequality, i.e., $KL(p,r)\leq KL(p,q)+KL(q,r)$ does not necessarily hold for distributions $p,q,r$, the length of the cycle that minimizes $\bI_{1,2}(\bH,\cA^{\text{exact}})$ is indeterministic. 
In the numerical studies of Section \ref{sec: numerical studies}, we will consider an example where the minimum Kullback-Leibler divergence of cycles of length $3$ is smaller than that of cycles of length $2$.


\subsection{Known lower bounds}
The third example considers the case of known lower bound on the number of streams following each hypothesis, i.e., $\cA=\cA^{\text{lower}}_{L_1,\ldots,L_M}$ in \eqref{def, known lower}.
Throughout this subsection, we treat $L_1,\cdots,L_M$ as fixed and use $\cA^{\text{lower}}$ to denote the prior information. 
We also fix arbitrary $\bH\in\cA^{\text{lower}}$ and $i\neq j\in[M]$.
For simplicity of presentation, we denote
\begin{equation*}
\begin{aligned}
    U(\bH) & := \{l\in[M]: |H_l|>L_l\}, \\
    V(\bH) & := \{l\in[M]: |H_i|=L_l\} = [M]\backslash U(\bH),
\end{aligned}
\end{equation*}
i.e., those hypotheses for which the number of streams following them is greater than the lower bound or equal to the lower bound, respectively. 

Let us consider how $\tilde\cA_{i,j}(\bH,\cA^{\text{lower}})$ is composed.
Note that, if $|H_i|>L_i$, then to make the fewest changes, it suffices to move one stream in $H_i$ to $H_j$, i.e., 
\begin{equation*}
    \widetilde\Alt_{i,j}(\bH,\cA^{\text{lower}}) = \{i\to j\}_\bH.
\end{equation*}
However, if $|H_i|=L_i$, after moving one stream in $H_i$ to $H_j$, the number of streams following hypothesis $i$ is not enough, so there must exist one stream in $[K]\backslash H_i$ that is moved back to $H_i$.
If that stream is from $H_u$ for some $u\in\{j\}\cup U(\bH)$, then there is nothing else we need to change, as the number of streams following each hypothesis now is consistent with the prior information. 
However, if that stream is from $H_v$ for some $v\in V(\bH)\backslash\{i,j\}$, then this operation makes the number of streams following hypothesis $v$ not enough, so there must exist another stream in $[K]\backslash(H_i\cup H_v)$ to compensate for it.
There are still two possible ways, either from $H_u$ for some $u\in\{j\}\cup U(\bH)$, or from $H_{v'}$ for some $v'\in V(\bH)\backslash\{i,j,v\}$.
Repeating this derivation, we end up with a backward chain of errors,
$j\leftarrow i\leftarrow v_m\leftarrow \cdots\leftarrow v_1\leftarrow u$ for some $v_1,\ldots,v_m\in V(\bH)\backslash\{i,j\}$ and $u\in\{j\}\cup U(\bH)$.
Note that when $u=j$, this chain is a loop, whereas when $u\neq j$, it is not. 
Formally, we have
\begin{equation*}
\begin{aligned}
    & \widetilde\Alt_{i,j}(\bH,\cA^{\text{lower}}) 
    = \bigcup_{\substack{ m\geq 0, \, u\in\{j\}\cup U(\bH) \\ \text{distinct } v_1,\ldots,v_m\in V(\bH)\backslash\{i,j\} }} \{u\to v_1\to\cdots\to v_m\to i\to j\}_\bH,
\end{aligned}    
\end{equation*}
whose size is also upper bounded by $K^M$.
Besides,
the constant in the asymptotic approximation, i.e., $\bI_{i,j}(\bH,\cA^{\text{lower}})$, is $\bI_{i,j}(\bH)$ if $|H_i|>L_i$, and 
\begin{equation*}
\begin{aligned}
    \min_{\substack{ m\geq 0, u\in\{j\}\cup U(\bH) \\ \text{distinct } v_1,\ldots,v_m\in V(\bH)\backslash\{i,j\} }} \Big\{ \bI_{u,v_1}(\bH) + \sum_{l=1}^{m-1} \bI_{v_l,v_{l+1}}(\bH) + \bI_{v_m,i}(\bH) \Big\} + \bI_{i,j}(\bH)
\end{aligned}
\end{equation*}
if $|H_i|=L_i$.

Moreover, the generic stopping time in \eqref{the proposed procedure} in this case is equivalent to 
\begin{equation*}
\begin{aligned}
    \hat T = & \inf_{n\geq 1} \Big\{ |\hat H_i(n)|\geq L_i \,\,\forall\,i\in[M], \text{ and} \\
    & \text{for all } m\geq 0, \text{ and distinct }
    u\in U(\hat\bH(n)), w\in [M], \text{ and } v_1,\ldots,v_m\in V(\hat\bH(n)), \\
    & \lambda_{u,v_1}(n)+\lambda_{v_1,v_2}(n)+\cdots+\lambda_{v_m,w}(n) \geq \max\{a_{v_1,u},a_{v_2,v_1},\ldots,v_{w,v_m}\}
    \Big\}.
\end{aligned}
\end{equation*}
In particular, when $M=2$, we have
\begin{equation*}
\begin{aligned}
    \hat T & = \hat T_{\emptyset} \wedge \hat T_1 \wedge \hat T_2 \wedge \hat T_{1,2},
\end{aligned}
\end{equation*}
where 
the subscript represents $\{i\in[M]:|\hat H_i(\cdot)|=L_i\}$, and
\begin{equation*}
\begin{aligned}
    \hat T_\emptyset = & \inf_{n\geq 1} \midbig\{ |\hat H_1(n)|>L_1, \, |\hat H_2(n)|>L_2, \, \lambda_{1,2}(n) \geq a_{2,1}, \, \lambda_{2,1}(n) \geq a_{1,2} \midbig\}, \\
    \hat T_1 = & \inf_{n\geq 1} \midbig\{ |\hat H_1(n)|=L_1, \, |\hat H_2(n)|>L_2, \, \lambda_{1,2}(n)+\lambda_{2,1}(n) \geq a_{2,1}\vee a_{1,2}, \, \lambda_{2,1}(n) \geq a_{1,2}
    \midbig\}, \\
    \hat T_2 = & \text{ analogous to $\hat T_1$}, \\
    \hat T_{1,2} = & \inf_{n\geq 1} \midbig\{ |\hat H_1(n)|=L_1, \, |\hat H_2(n)|=L_2, \, \lambda_{1,2}(n)+\lambda_{2,1}(n) \geq a_{2,1}\vee a_{1,2}
    \midbig\}.
\end{aligned}
\end{equation*}
Note that $\hat T_{1,2}$ occurs if and only if $L_1+L_2=K$, i.e., the prior information is actually known numbers.
Also note that this recovers the ``gap-intersection rule" in \cite{Song_prior}.
When $M=3$, we have
\begin{equation*}
    \hat T = \hat T_\emptyset \wedge \min_{i\in\{1,2,3\}} \hat T_i \wedge \min_{(i,j)\in\{(1,2),(2,3),(3,1)\}} \hat T_{i,j} \wedge \hat T_{1,2,3},
\end{equation*}
where 
\begin{equation*}
\begin{aligned}
    \hat T_\emptyset = & \inf_{n\geq 1} \midbig\{ |\hat H_1(n)|>L_1, |\hat H_2(n)|>L_2, |\hat H_3(n)|>L_3, \\
    & \lambda_{1,2}(n) \geq a_{2,1}, \, \lambda_{1,3}(n) \geq a_{3,1}, \,
    \lambda_{2,1}(n) \geq a_{1,2}, \, \lambda_{2,3}(n) \geq a_{3,2}, \,
    \lambda_{3,1}(n) \geq a_{1,3}, \, \lambda_{3,2}(n) \geq a_{2,3}
    \midbig\}, \\
    \hat T_1 = & \inf_{n\geq 1} \midbig\{ |\hat H_1(n)|=L_1, |\hat H_2(n)|>L_2, |\hat H_3(n)|>L_3, \\
    & \lambda_{2,1}(n)+\lambda_{1,2}(n)\geq a_{1,2}\vee a_{2,1}, \,
    \lambda_{3,1}(n)+\lambda_{1,2}(n)\geq a_{1,3}\vee a_{2,1}, \\
    & \lambda_{2,1}(n)+\lambda_{1,3}(n)\geq a_{1,2}\vee a_{3,1}, \,
    \lambda_{3,1}(n)+\lambda_{1,3}(n)\geq a_{1,3}\vee a_{3,1}, \\
    & \lambda_{2,1}(n) \geq a_{1,2}, \, \lambda_{2,3}(n) \geq a_{3,2}, \,
    \lambda_{3,1}(n) \geq a_{1,3}, \, \lambda_{3,2}(n) \geq a_{2,3}
    \midbig\}, \\
    \hat T_2, \hat T_3 = & \text{ analogous to $\hat T_1$}, \\
    \hat T_{1,2} = & \inf_{n\geq 1}\midbig\{ |\hat H_1(n)|=L_1, |\hat H_2(n)|=L_2, |\hat H_3(n)|>L_3, \\
    & \lambda_{2,1}(n)+\lambda_{1,2}(n)\geq a_{1,2}\vee a_{2,1}, \,
    \lambda_{3,1}(n)+\lambda_{1,2}(n)\geq a_{1,3}\vee a_{2,1}, \\
    & \lambda_{3,2}(n)+\lambda_{2,1}(n)+\lambda_{1,3}(n)\geq a_{2,3}\vee a_{1,2}\vee a_{3,1}, \,
    \lambda_{3,1}(n)+\lambda_{1,3}(n)\geq a_{1,3}\vee a_{3,1}, \\
    & \lambda_{3,2}(n)+\lambda_{2,1}(n)\geq a_{2,3}\vee a_{1,2}, \\
    & \lambda_{3,1}(n)+\lambda_{1,2}(n)+\lambda_{2,3}(n)\geq a_{1,3}\vee a_{2,1}\vee a_{3,2}, \,
    \lambda_{3,2}(n)+\lambda_{2,3}(n)\geq a_{2,3}\vee a_{3,2}, \\
    & \lambda_{3,1}(n)\geq a_{1,3}, \, \lambda_{3,2}(n)\geq a_{2,3}
    \midbig\}, \\
    \hat T_{1,3}, \hat T_{2,3} = & \text{ analogous to $\hat T_{1,2}$}, \\
    \hat T_{1,2,3} = & \text{ the stopping time in \eqref{hat T, known numbers and M=3}}.
\end{aligned}
\end{equation*}
Note that in the above expressions, some elements in the maximums of thresholds can be omitted,
e.g., in the last two lines of $\hat T_{1,2}$, $\lambda_{3,2}(n)\geq a_{2,3}$ implies $\lambda_{3,2}(n)+\lambda_{2,3}(n)\geq a_{2,3}$, so it suffices to have $\lambda_{3,2}(n)+\lambda_{2,3}(n)\geq a_{3,2}$ in the second last line, but we do not omit them for clarity.
Meanwhile, in each of the stopping times, all inequalities do not imply each other, so all criteria need to be checked. 
This is substantially more complex than the case of no prior information or known numbers, and than the case of two candidate hypotheses for each stream.

\subsection{Exclusive hypotheses}
The fourth example considers the case of exclusive hypotheses, i.e., $\cA=\cA^{\text{exclu}}_{i,j}$ in \eqref{def, exclu hyps} for some $i\neq j\in[M]$.
First note that, when $M=2$ this reduces to a binary testing problem of testing ``all streams follow hypothesis $1$" versus ``all streams follow hypothesis $2$", so we focus on $M\geq 3$ and, without loss of generality, throughout this subsection we assume that it is hypothesis $1$ and hypothesis $2$ that exclude each other, and denote the prior information as $\cA^{\text{exclu}}_{1,2}$.

Let us study how $\widetilde\Alt_{i,j}(\bH,\cA^{\text{exclu}}_{1,2})$, where $\bH\in\cA^{\text{exclu}}_{1,2}$ and $i\neq j\in[M]$, is composed.
According to the definition of $\cA^{\text{exclu}}_{1,2}$, there are the following three cases regarding $H_1$ and $H_2$, which we consider separately.

(I) $H_1=H_2=\emptyset$. In this case, it suffices to consider $i\in[M]\backslash\{1,2\}$ and $j\in[M]\backslash\{i\}$, and for any such indices, we have
\begin{equation*}
    \widetilde\Alt_{i,j}(\bH,\cA^{\text{exclu}}_{1,2}) = \{i\to j\}_\bH,
\end{equation*}
and, thus,
\begin{equation*}
    \bI_{i,j}(\bH,\cA^{\text{exclu}}_{1,2}) = \bI_{i,j}(\bH).
\end{equation*}

(II) $H_1\neq\emptyset$ and $H_2=\emptyset$. 
In this case, a key observation is that, because $H_1$ and $H_2$ cannot be non-empty simultaneously, moving any stream to $H_2$ requires moving all streams in $H_1$ out. 
We denote the latter by $\{\forall\,1\to [M]\backslash\{1\}\}_\bH$, and separately consider the following three cases regarding $i\in[M]\backslash\{2\}$ and $j\in[M]\backslash\{i\}$.

\noindent (II.i) $i\in[M]\backslash \{2\}$ and $j\in[M]\backslash\{i,2\}$. Moving a stream in $H_i$ to $H_j$ for such $i,j$ does not violate the prior information, so
\begin{equation*}
\begin{aligned}
    \widetilde\Alt_{i,j}(\bH,\cA^{\text{exclu}}_{1,2}) & = \{i\to j\}_\bH, \\
    \bI_{i,j}(\bH,\cA^{\text{exclu}}_{1,2}) & = \bI_{i,j}(\bH).
\end{aligned}
\end{equation*}

\noindent (II.ii) $i=1$ and $j=2$. After moving one stream from $H_1$ to $H_2$, all streams in $H_1$ have to be moved out, so
\begin{equation*}
\begin{aligned}
    \widetilde\Alt_{1,2}(\bH,\cA^{\text{exclu}}_{1,2}) & = \{1\to 2\}_\bH \cap \{\forall\,1\to [M]\backslash\{1\}\}_\bH, \\
    \bI_{1,2}(\bH,\cA^{\text{exclu}}_{1,2}) & = \min_{k\in H_1} \Big\{ I^k_{1,2} + \sum_{k'\in H_1\backslash\{k\}} \min_{l\in[M]\backslash\{1\}} I^{k'}_{1,l} \Big\}.
\end{aligned}
\end{equation*}

\noindent (II.iii) $i\in[M]\backslash\{1,2\}$ and $j=2$. Similar to (II.ii), after moving one stream from $H_i$ to $H_2$, all streams in $H_1$ have to be moved out, so 
\begin{equation*}
\begin{aligned}
    \widetilde\Alt_{i,2}(\bH,\cA^{\text{exclu}}_{1,2}) & = \{i\to 2\}_\bH \cap \{\forall\,1\to [M]\backslash\{1\}\}_\bH, \\
    \bI_{i,2}(\bH,\cA^{\text{exclu}}_{1,2}) & = \bI_{i,2}(\bH) + \sum_{k\in H_1} \min_{l\in[M]\backslash\{1\}} I^k_{1,l}.
\end{aligned}
\end{equation*}

(III) $H_1=\emptyset$ and $H_2\neq\emptyset$. 
This is analogous to (II) and omitted.

Moreover, the generic stopping time in \eqref{the proposed procedure} in this case is equivalent to 
\begin{equation*}
    \hat T = \hat T_\emptyset \wedge \hat T_1 \wedge \hat T_2,
\end{equation*}
where the subscript represents $\{i\in\{1,2\}:\hat H_i(\cdot)\neq\emptyset\}$, and 
\begin{equation*}
\begin{aligned}
    \hat T_\emptyset & = \inf_{n\geq 1} \midbig\{ \hat H_1(n)=\emptyset, \, \hat H_2(n)=\emptyset, \text{ and } 
    \lambda_{i,j}(n) \geq a_{j,i} \text{ for all } i\in[M]\backslash\{1,2\} \text{ and } j\in[M]\backslash\{i\} \midbig\},
\end{aligned}
\end{equation*}
\begin{equation*}
\begin{aligned}
    \hat T_1 & = \inf_{n\geq 1} \Big\{ \hat H_1(n)\neq\emptyset, \, \hat H_2(n) = \emptyset, \text{ and the following three events occur:} \\
    & (i) \, \lambda_{i,j}(n) \geq a_{j,i} \text{ for all } i\in[M]\backslash\{2\} \text{ and } j\in[M]\backslash\{i,2\}, \\
    & (ii) \, \min_{k\in\hat H_1(n)} \Big\{ \midbig( \ell^k_1(n) - \ell^k_2(n) \midbig) +
    \sum_{k'\in\hat H_1(n)\backslash\{k\}} \min_{l\in[M]\backslash\{1\}} \midbig\{ \ell^k_1(n) - \ell^k_l(n) \midbig\} \Big\} \geq a_{2,1}, \\
    & (iii) \, \forall\, i\in[M]\backslash\{1,2\}, \,
    \lambda_{i,2}(n) + \sum_{k\in\hat H_1(n)} \min_{l\in[M]\backslash\{1\}} \midbig\{ \ell^k_1(n) - \ell^k_l(n) \midbig\} \geq a_{2,i}
    \Big\},
\end{aligned}
\end{equation*}
and $\hat T_2$ is analogous to $\hat T_1$ and thus omitted. 

Based on the analysis of the examples above, we can see that, by identifying the minimal subsets of alternative hypothesis configurations, i.e., $\widetilde\Alt_{i,j}(\bH,\cA)$, we are able to quickly understand the essence of the problem and design an efficient procedure.


\section{Numerical studies} \label{sec: numerical studies}
In this section, we present numerical studies. 
Since the proposed procedure in \eqref{the proposed procedure} is the first to address the problem of sequential multiple testing with multiple hypotheses, 
and it reduces to the state-of-the-art procedures in \cite{Song_prior} directly and to those in \cite{Kobi_2015_AO, Kobi_2020_composite, Aris_IEEE, PaperIII} with corresponding modifications, in the case of two hypotheses, there is no suitable benchmark for comparison.
Thus, the main purpose of the numerical studies is to illustrate the implementation details and properties of the proposed procedure.
In Subsection \ref{subsec: study 1}, we introduce the setup and compare the expected sample size of the proposed procedure under different sets of prior information, which illustrates the benefits of prior information and the asymptotic optimality theory. 
In Subsection \ref{subsec: study 2}, we discuss how to use importance sampling to estimate the actual error probabilities of the proposed procedure, and demonstrate its reliability and its loss of reliability if part of its stopping criteria are neglected. 

\subsection{Setup and study 1} \label{subsec: study 1}
We consider $K=3$ streams of i.i.d. Gaussian data with means $\btheta=(\theta^1,\theta^2,\theta^3)$ and unit variances.
There are $M=3$ hypotheses for each stream:
\begin{equation*}
\begin{aligned}
    \Theta^1 & =\{\theta^1_1,\theta^1_2,\theta^1_3\}=\{0,1,-0.5\}, \\
    \Theta^2 & =\{\theta^2_1,\theta^2_2,\theta^2_3\}=\{-0.5,0,1\}, \\
    \Theta^3 & =\{\theta^3_1,\theta^3_2,\theta^3_3\}=\{1,-0.5,0\}.
\end{aligned}
\end{equation*}
We set the true global parameter as $\btheta=(0,0,0)$.
Based on our notations, $\bH$ is equal to
$(H^1,H^2,H^3)=(1,2,3)$ and $(H_1,H_2,H_3)=(\{1\},\{2\},\{3\})$.
We consider the following four cases of prior information, all of which admit this global parameter: 
(i) no prior information, 
(ii) weak lower bounds $\cA^{\text{lower}}_{1,0,0}$,
(iii) strong lower bounds $\cA^{\text{lower}}_{1,1,0}$, and (iv) known numbers $\cA^{\text{exact}}_{1,1,1}$.
Note that the four cases of prior information are from weaker to stronger.
Besides, we choose equal thresholds $a_{i,j}=a$ for all $i\neq j\in[3]$.
We focus on this relatively simple setup in order to clearly convey the key ideas and avoid unnecessary distractions.
We arrange the hypotheses and the true global parameter in this way so that the sum of Kullback-Leibler divergences of a cycle of length $3$ is smaller than that of all pairs, which will be computed explicitly next.

\begin{table}[]
    \centering
    \begin{tabular}{|c|c|c|c|c|c|c|c|}
        \hline
         $\cA$ & $\bI_{1,2}(\bH,\cA)$ & $\bI_{1,3}(\bH,\cA)$ & $\bI_{2,1}(\bH,\cA)$ & $\bI_{2,3}(\bH,\cA)$ & $\bI_{3,1}(\bH,\cA)$ & $\bI_{3,2}(\bH,\cA)$ & MIN  \\
        \hline
        $[M]^K$ & $1/2$ & $1/8$ & $1/8$ & $1/2$ & $1/2$ & $1/8$ & $1/8$ \\
        \hline
        $\cA^{\text{lower}}_{1,0,0}$ & $5/8$ & $1/4$ & $1/8$ & $1/2$ & $1/2$ & $1/8$ & $1/8$ \\
        \hline
        $\cA^{\text{lower}}_{1,1,0}$ & $5/8$ & $3/8$ & $1/4$ & $5/8$ & $1/2$ & $1/8$ & $1/8$ \\
        \hline
        $\cA^{\text{exact}}_{1,1,1}$ & $5/8$ & $3/8$ & $3/8$ & $5/8$ & $5/8$ & $3/8$ & $3/8$ \\
        \hline
    \end{tabular}
    \caption{Minimum Kullback-Leibler divergences between $\bH$ and all alternative sets in the four cases of prior information.}
    \label{table of Is}
\end{table}

Note that the Kullback-Leibler divergence between two Gaussian distributions with means $\theta$ and $v$ and unit variances is $(\theta-v)^2/2$, so we have
$I^1_{1,2} = 1/2$,
$I^1_{1,3} = 1/8$, 
$I^2_{2,1} = 1/8$, 
$I^2_{2,3} = 1/2$, 
$I^3_{3,1} = 1/2$, 
$I^3_{3,2} = 1/8$.
The minimum Kullback-Leibler divergence between $\bH$ and $\bI_{i,j}(\bH,\cA)$ for the four cases of $\cA$ and the six pairs of $i,j$ are summarized in Table \ref{table of Is}.
The last column represents $\min_{i\neq j\in[M]} \bI_{i,j}(\bH,\cA)$. 
In particular, when $\cA=\cA^{\text{exact}}_{1,1,1}$, the three pairwise sums and the two cyclic sums of Kullback-Leibler divergences are 
\begin{equation*}
\begin{gathered}
    \bI_{1,2}(\bH)+\bI_{2,1}(\bH) = 5/8, \quad
    \bI_{1,3}(\bH)+\bI_{3,1}(\bH) = 5/8, \quad 
    \bI_{2,3}(\bH)+\bI_{3,2}(\bH) = 5/8, \\
    \bI_{1,2}(\bH)+\bI_{2,3}(\bH)+\bI_{3,1}(\bH) = 3/2, \quad
    \bI_{1,3}(\bH)+\bI_{3,2}(\bH)+\bI_{2,1}(\bH) = 3/8,
\end{gathered}
\end{equation*}
rendering $\bI_{i,j}(\bH,\cA^{\text{exact}}_{1,1,1})=\bI_{i,j}^{\text{len}=3}(\bH,\cA^{\text{exact}}_{1,1,1})$ for $(i,j)=(1,3), (3,2)$ and $(2,1)$.
Based on the asymptotic approximation in Corollary \ref{corollary, AO}, we know that, as $a\to\infty$, the asymptotic approximations to the expected sample size of the proposed procedure in the four cases of prior information are $8a$, $8a$, $8a$, $8a/3$, respectively. 
Note that prior information of known numbers reduces the asymptotic approximation remarkably, as it allows summing evidence across streams, whereas prior information of known lower bounds may not.
However, when the thresholds, i.e., error tolerances, are not all equal, the benefit of known lower bounds can be as strong as that of known numbers,
e.g., when $\cA=\cA^{\text{lower}}_{1,1,0},$ $a_{2,1}=a_{3,1}=a_1$, $a_{1,2}=a_{3,2}=a_2$ and $a_{1,3}=a_{2,3}=a_3$, we have $\bExp_\btheta[\hat T] \sim \max\{ 8a_1/3, 4a_2, 8a_3 \}$, which can be as small as $8a/3$ if $a_1=a$, $a_2=2a/3$ and $a_3=a/3$.

\begin{figure*} 
    \centering
    \includegraphics[width=0.49\textwidth]{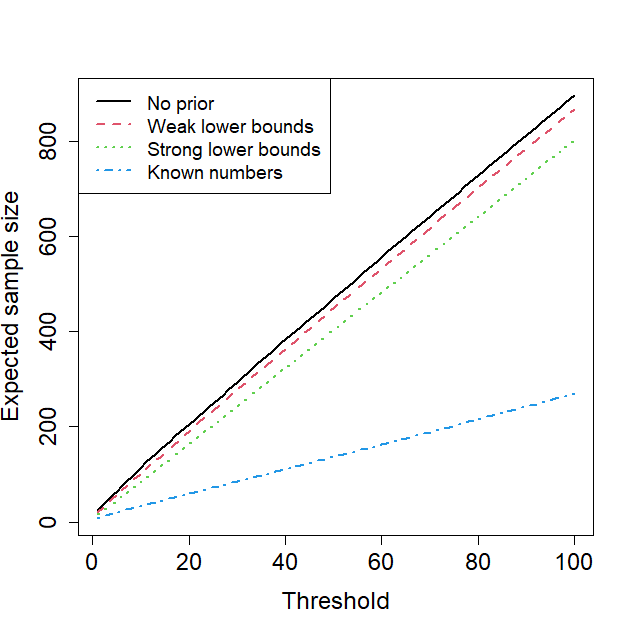} 
    \includegraphics[width=0.49\textwidth]{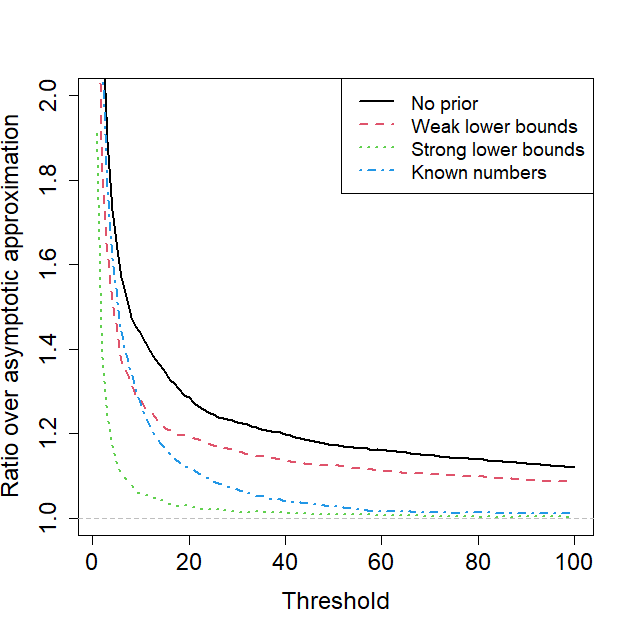}
    \caption{
    Expected sample sizes (left) and ratios of expected
sample sizes to their asymptotic approximations (right) of the
proposed procedure against equal thresholds.
The four curves, viewed from top to bottom on the far right of the first sub-figure, correspond to no prior, weak lower bounds, strong lower bounds, and known numbers, respectively.
Simulations are based on $10^4$ runs. 
    }
    \label{Fig1}
\end{figure*}

In Figure \ref{Fig1}, we plot the expected sample size of the proposed procedure and its ratio against the asymptotic approximation, in each of the four cases of prior information, against threshold. 
We can see that the expected sample sizes are basically linear in the thresholds, and their ratios with the asymptotic approximations converge to one, consistent with the asymptotic theory.
Moreover,
the order of the expected sample sizes in the four cases of prior information is 
$\cA^{\text{exact}}_{1,1,1} \leq \cA^{\text{lower}}_{1,1,0} \leq \cA^{\text{lower}}_{1,0,0} \leq [M]^K$,
i.e.,
stronger prior information leads to a smaller expected sample size, which is consistent with intuition as stronger prior information reduces the number of alternative hypothesis configurations.

\subsection{Study 2} \label{subsec: study 2}
In this subsection, we focus on the prior information of known numbers, i.e., $\cA^{\text{exact}}_{1,1,1}$. 
We first discuss how to apply importance sampling to estimate the actual error probabilities, based on which we are able to demonstrate the error control of the proposed procedure and its failure of error control if part of its stopping criteria are neglected. 

The basic idea of importance sampling, according to, e.g., \cite{Siegmumd_IS, Bucklew_Book, Song_prior, Song_2025Efficient}, is that, in order to evaluate a small $\Pro(\Gamma)$, we find another probability measure $\Qro$ such that
(i) $\Pro$ and $\Qro$ are mutually absolutely continuous, so that the likelihood ratio $d\Qro/d\Pro$ is well-defined,
(ii) $\Qro(\Gamma)$ is not small, so that event $\Gamma$ can be observed frequently under $\Qro$, and
(iii) $\Qro$ is similar to $\Pro$, so that the variance of the importance weight is not too large.
Then, based on the following Wald's likelihood ratio identity:
\begin{equation*}
    \Pro(\Gamma) = \Exp_\Pro[1\{\Gamma\}] = \Exp_\Qro\left[\left(\frac{d\Qro}{d\Pro}\right)^{-1}1\{\Gamma\}\right],
\end{equation*}
one can estimate $\Pro(\Gamma)$ by simulating $(d\Qro/d\Pro)^{-1} 1\{\Gamma\}$ under $\Qro$ for many times and take the average. 
Specifically, in this work, given prior information $\cA$, configuration $\bH$, and hypotheses $i\neq j\in[M]$, in order to estimate the probability of making type-$(i,j)$ errors, i.e., $\bPro_\bH(H_i\cap D_j\neq\emptyset)$, we recommend using importance distribution
\begin{equation*}
    \bPro_\bH^*:=\frac{1}{|\widetilde\Alt_{i,j}(\bH,\cA)|} \sum_{\bA\in\widetilde\Alt_{i,j}(\bH,\cA)} \bPro_\bA,
\end{equation*}
where we suppress its dependence on $\cA$ and $i,j$ to lighten the notations.
It is clear that 
(i) $\bPro_\bH$ and $\bPro_\bH^*$ are mutually absolutely continuous, 
(ii) $\bPro_\bH^*(H_i\cap D_j\neq\emptyset)\geq 1-\alpha_{\operatorname{sum}}$ is large, since for any $\bA\in\widetilde\Alt_{i,j}(\bH,\cA)$ we have $\bPro_\bA(H_i\cap D_j\neq\emptyset)\geq\bPro_\bA(\bD=\bA)\geq 1-\alpha_{\operatorname{sum}}$,
and 
(iii) $\bPro_\bH^*$ is similar to $\bPro_\bH$ among distributions that satisfy (ii).
Besides, the likelihood ratio between $\bPro_\bH^*$ and $\bPro_\bH$ based on the data up to time $n$ is 
\begin{equation*}
\begin{aligned}
    \frac{d\bPro_\bH^*}{d\bPro_\bH}(\bcF(n)) & = \frac{1}{|\widetilde\Alt_{i,j}(\bH,\cA)|}\sum_{\bA\in\widetilde\Alt_{i,j}(\bH,\cA)} \exp\{\bell_\bA(n)-\bell_\bH(n)\} \\
    & = \frac{1}{|\widetilde\Alt_{i,j}(\bH,\cA)|}\sum_{\bA\in\widetilde\Alt_{i,j}(\bH,\cA)} \exp\left\{\sum_{k\in\bH\triangle\bA} \left(\ell^k_{A^k}(n)-\ell^k_{H^k}(n)\right) \right\}.
\end{aligned}
\end{equation*}

In this numerical study, we focus on the probability of making type-$(1,3)$ errors, i.e., 
\begin{equation*}
    \bPro_\bH(H_1\cap D_3\neq\emptyset)=\bPro_\bH\midbig((D^1,D^2,D^3)=(3,2,1) \text{ or } (3,1,2)\midbig), 
\end{equation*}
where $(3,2,1)$ makes a pair of errors and $(3,1,2)$ makes a cycle of errors. 
We estimate this probability not only when the proposed procedure in \eqref{hat T, known numbers and M=3} is used correctly, but also when part of its stopping criteria are neglected.
Specifically, we consider two scenarios: when its first three criteria against pairwise errors are neglected, 
and when its last two criteria against cyclic errors are neglected.
In the left of Figure \ref{Fig2}, we plot the actual error probabilities of the three procedures along with the theoretical upper bound on the error probability if the procedure is used correctly. 
We can see that, when the procedure is used correctly, its actual error probability is well-controlled by the nominal level, and the extent of conservativeness is low when the number of streams is small;  
but when the procedure is used incorrectly, its actual error probability explodes.
In the right of Figure \ref{Fig2}, we plot the relative errors of the estimates.
We can see that the relative errors in estimating probabilities as small as $10^{-40}$ are below $2.5\%$ based on only $10^4$ runs, demonstrating the high efficiency of the proposed importance sampling approach. 

\begin{figure*} 
    \centering
    \includegraphics[width=0.49\textwidth]{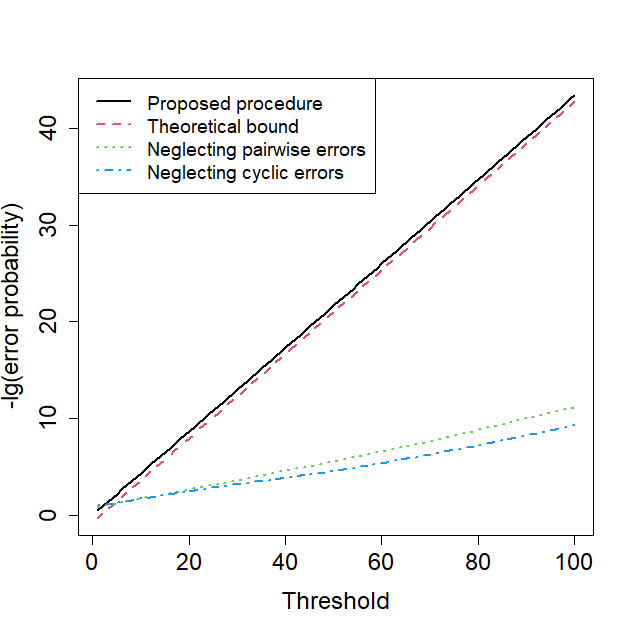} 
    \includegraphics[width=0.49\textwidth]{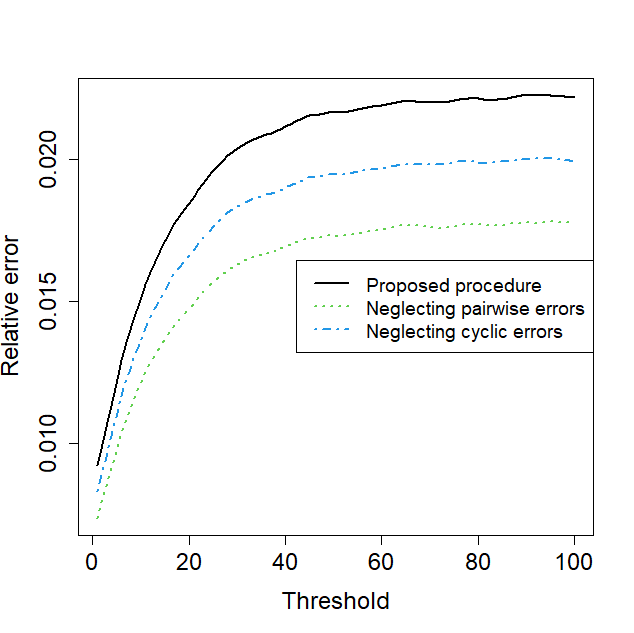}
    \caption{Negative base-$10$ logarithm of the error probabilities (left) and the relative errors of the estimation (right). 
    The four curves, viewed from top to bottom on the far right of the first sub-figure, correspond to the proposed procedure in \eqref{hat T, known numbers and M=3}, the theoretical upper bound on its error probability in Theorem \ref{theorem, error control}, the proposed procedure in \eqref{hat T, known numbers and M=3} with the first three criteria neglected, and the proposed procedure in \eqref{hat T, known numbers and M=3} with the last two criteria neglected.
    Simulations are based on $10^4$ runs. 
    }
    \label{Fig2}
\end{figure*}

\section{Extension to composite hypotheses} \label{sec: extension to composite hypotheses}
In this section, we consider the case where the hypotheses are not simple, but composite and parametrized.
We start with a brief problem formulation that highlights the difference from the case of simple hypotheses.

Similarly as before, suppose that $X^k$, $k\in[K]$ are $K$ independent streams of i.i.d. data. 
But differently, for every $k\in[K]$, suppose that the density of $X^k(1)$ is given by $f^k_{\theta^k}$ where $\theta^k\in\Theta^k$ is an unknown parameter, and consider the following hypothesis testing problem:
\begin{equation*}
    \theta^k\in\Theta^k_i \text{ for } i\in[M],
\end{equation*}
where $\{\Theta^k_i:i\in[M]\}$ are disjoint and form a partition of $\Theta^k$. 
We denote by $\btheta:=(\theta^1,\ldots,\theta^K)\in\bTheta:=\Theta^1\otimes\cdots\otimes\Theta^K$ the global parameter, and by $\bH(\btheta):=(H^1(\btheta),\ldots,H^K(\btheta)):=(H_1(\btheta),\ldots,H_M(\btheta))$ the corresponding hypothesis configuration. 
We denote by $\bPro_\btheta=\prod_{k\in[K]} \Pro^k_{\theta^k}$ the global distribution and $\bExp_\btheta$ and $\Exp^k_{\theta^k}$, $k\in[K]$ the corresponding expectations.
The data filtration and a testing procedure are defined the same as before.
For any prior information $\cA\subseteq[M]^K$ about the unknown configuration $\bH(\btheta)$, and error levels $\balpha$, the subfamily of reliable procedures is defined as
\begin{equation*}
\begin{aligned}
    \Delta(\balpha,\cA) := \{
    (T,\bD)\in\Delta: \, & \bPro_\btheta(T<\infty) = 1 \text{ and} \\
    & \bPro_\btheta(H_i(\btheta)\cap D_j\neq\emptyset) \leq \alpha_{i,j} \text{ for all } \btheta\in\bTheta_\cA \text{ and } i\neq j\in[M]
    \},
\end{aligned}
\end{equation*}
where $\bTheta_\cA:=\{\btheta\in\bTheta: \bH(\btheta)\in\cA\}$ denotes all global parameter values consistent with the prior information.
Our goal is to achieve 
\begin{equation*}
    \cL_\btheta(\balpha,\cA) := \inf\{ \bExp_\btheta[T]: (T,\bD)\in\Delta(\balpha,\cA) \},
\end{equation*}
simultaneously for all $\btheta\in\bTheta_\cA$ as $\alpha_\max\to 0$.

\subsection{Universal lower bound}
We assume that for every $k\in[K]$ and $\theta^k\neq v^k\in\Theta^k$, the Kullback-Leibler divergence between $\theta^k$ and $v^k$ are positive and finite, i.e., 
\begin{equation*}
    I^k_{\theta^k,v^k} := \int f^k_{\theta^k} \log\frac{f^k_{\theta^k}}{f^k_{v^k}} d\nu^k\in (0,\infty),
\end{equation*}
and that for every $j\in[M]$ that $\theta^k\notin\Theta^k_j$, the minimum Kullback-Leibler divergence between $\theta^k$ and $\Theta^k_j$ is positive, i.e., 
\begin{equation*}
    I^k_{\theta^k}(j) := \inf_{v^k\in\Theta^k_j} I^k_{\theta^k,v^k} > 0.
\end{equation*}
We define $\Alt_{i,j}(\bH(\btheta),\cA)$ and $\widetilde\Alt_{i,j}(\bH(\btheta),\cA)$ the same as before, and define
\begin{equation*}
\begin{aligned}
    \bI_{i,j}(\btheta,\cA) & := \min_{\bA\in\widetilde\Alt_{i,j}(\bH(\btheta),\cA)} \bI_\btheta(\bA), \\
    \text{where } \bI_\btheta(\bA) & := 
    \sum_{k\in\bH(\btheta)\triangle\bA} I^k_{\theta^k}(A^k).
\end{aligned}
\end{equation*}

We are now ready to state the following universal lower bound.
All proofs in this section are deferred to Appendix \ref{appendix: proofs related to composite hyps}.
\begin{theorem} \label{theorem, composite, universal lower bound}
    For any prior information $\cA$, global parameter $\btheta\in\bTheta_\cA$, and error probabilities $\balpha$, we have
    \begin{equation*}
        \cL_\btheta(\balpha,\cA) \geq \max_{i\neq j\in[M]} \frac{\varphi(\alpha_{\operatorname{sum}},\alpha_{j,i})}{\bI_{i,j}(\btheta,\cA)}.
    \end{equation*}
\end{theorem}


\subsection{Methodology and analysis}
For any stream $k\in[K]$, hypothesis $i\in[M]$, and time $n\geq 1$, we denote by $\ell^k_i(n)$ the maximum log-likelihood in stream $k$ for hypothesis $i$ based on the first $n$ data, i.e., 
\begin{equation*}
    \ell^k_i(n) := \sup_{\theta^k\in\Theta^k_i} \sum_{t=1}^n \log f^k_{\theta^k}(X^k(t)), 
\end{equation*}
and for any configuration $\bH\in[M]^K$, denote 
\begin{equation*}
    \bell_\bH(n):=\sum_{k\in[K]} \ell^k_{H^k}(n).
\end{equation*}
The maximum likelihood hypothesis $\hat H^k(n)$ for every $k\in[K]$ and the maximum likelihood configuration $\hat\bH(n)$ are defined the same as before.

In order to ensure that this procedure terminates almost surely and to establish an asymptotic upper bound on its expected sample size, we assume that for any $k\in[K]$, $\theta^k\in\Theta^k$, $j\in[M]$ that $\theta^k\notin\Theta^k_j$, and $\epsilon\in(0,1)$,
\begin{equation} \label{composite, complete assumption}
    \sum_{n=1}^\infty \Pro^k_{\theta^k} \Big( \frac{1}{n}\midbig( \ell^k_{\hat H^k(n)} - \ell^k_j(n) \midbig) \leq (1-\epsilon) I^k_{\theta^k}(j) \Big) < \infty.
\end{equation}
In order to ensure the error control, we assume that there exists a threshold function, $\beta(n,\alpha)$ for $n\geq 1$ and $\alpha\in(0,1)$, such that 
\begin{equation} \label{composite, rate of beta}
    \beta(n,\alpha) \leq |\log\alpha|(1+o(1)) + o(n) \text{ as } \alpha\to 0 \text{ and } n\to\infty,
\end{equation}
and, for any $\btheta\in\bTheta$, $\bA\neq\bH(\btheta)$, and $\alpha\in(0,1)$,
\begin{equation} \label{composite, ville-type assumption}
    \bPro_\btheta\Big( \exists\,n\geq 1, \bell_{\bA}(n) - \bell_{\bH(\btheta)}(n) \geq \beta(n,\alpha) \Big) \leq \alpha.
\end{equation}
Such threshold functions can be designed, e.g., following \cite{mixture2021} when $\Theta^k$, $k\in[K]$ are subsets of the parameter space of a one-dimensional exponential family, and following \cite{Ali2025} for more general distribution families.

Then, we are ready to present the proposed procedure:
For any prior information $\cA$ and error levels $\balpha$, 
\begin{equation} \label{composite, the proposed procedure}
\begin{aligned}
    \hat T & := \inf_{n\geq 1} \Big\{
    \hat\bH(n) \in \cA \text{ and } \bell_{\hat\bH(n)}(n) - \max_{\bA\in\widetilde\Alt_{i,j}(\hat\bH(n),\cA)} \bell_\bA(n) \geq \beta\midbig(n,\alpha_{j,i}/b_{j,i}(\cA)\midbig) \text{ for all } i\neq j\in[M]
    \Big\}, \nonumber \\
    \hat\bD & := \hat\bH(\hat T). 
\end{aligned}
\end{equation}
The following two theorems establish the error control and asymptotic upper bound on the expected sample size of this procedure. Combining them and Theorem \ref{theorem, composite, universal lower bound} completes the asymptotic optimality theory.

\begin{theorem} \label{theorem, composite, error control}
    For any prior information $\cA$ and error levels $\balpha$, we have $(\hat T,\hat\bD)\in\Delta(\balpha,\cA)$.
\end{theorem}

\begin{theorem} \label{theorem, composite, AUB}
    For any prior information $\cA$ and global parameter $\btheta\in\bTheta_\cA$, we have
    \begin{equation*}
        \bExp_\btheta[\hat T] \lesssim \max_{i\neq j\in[M]} \frac{|\log\alpha_{j,i}|}{\bI_{i,j}(\btheta,\cA)}
    \end{equation*}
    as $\alpha_\max\to 0$.
\end{theorem}



\section{Conclusion and discussions} \label{sec: conclusion}
In this work, we formulate the problem of sequential multiple testing with multiple hypotheses in each testing problem and prior information on the unknown hypothesis configuration. 
We design a testing procedure that is computationally much more efficient than the na\"ive procedure by concentrating on alternative hypothesis configurations that are the most adjacent to the maximum likelihood one. 
We demonstrate that this procedure is both reliable in controlling the error probabilities and asymptotically optimal in minimizing the expected sample size. These general results are specialized to four concrete examples of prior information and are readily applicable to others. 


Here are some directions to extend this work:
(i) error metrics other than the familywise error probabilities, such as (the multihypothesis version of) the false discovery rates,
following \cite{Bartroff_2021_EquivErrorMetrics}, 
(ii) asynchronous decisions, i.e., decisions for different streams are made at different times, following \cite{PaperIII}, 
and (iii) a second-order asymptotic analysis of the expected sample size, following \cite{song_higherorder}.
There are also open questions about the general problem of sequential multiple testing with multiple hypotheses, which have been discussed in the introduction, such as controlling generalized error metrics, incorporating sampling constraints, and handling dependence or hierarchical structures among streams.
We hope this work could
draw attention to
the non-triviality of extending from two to multiple hypotheses in the problem of sequential multiple testing. 

\appendices
\section{Proofs in Section \ref{sec: universal lower bound} and \ref{sec: methodology}} \label{appendix, proofs}
\begin{proof} [Proof of Theorem \ref{theorem, LB}]
    Fix arbitrary $\cA,\bH\in\cA$ and $\balpha$. Also fix arbitrary $(T,\bD)\in\Delta(\balpha,\cA)$.
    For any $i\neq j\in[M]$ and $\bA\in\Alt_{i,j}(\bH,\cA)$, by Wald's identity, we have
    \begin{equation*}
    \begin{aligned}
        & \, \bExp_\bH\left[ \bell_\bH(T) - \bell_\bA(T) \right] = \bExp_\bH[T] \, \bExp_\bH\left[ \bell_\bH(1) - \bell_\bA(1) \right] = \bExp_\bH[T] \, \bI_\bH(\bA).
    \end{aligned}
    \end{equation*}
    Meanwhile, by the information-theoretical inequality (see, e.g., \cite[Lemma 3.2.1]{Tartakovsky_Book}) and the fact that $\bD$ is $\bcF(T)$-measurable, we have
    \begin{equation*}
    \begin{aligned}
        \bExp_\bH\left[ \bell_\bH(T) - \bell_\bA(T) \right] \geq \varphi\left( \bPro_\bH(\bD\neq\bH), \, \bPro_\bA(\bD=\bH) \right).
    \end{aligned}
    \end{equation*}
    Since $(T,\bD)\in\Delta(\alpha,\bH)$, $\bH,\bA\in\cA$, and the fact that $\bA\in\Alt_{i,j}(\bH,\cA)$ is equivalent to $\bH\in\Alt_{j,i}(\bA,\cA)$, 
    we have 
    $\bPro_\bH(\bD\neq\bH)\leq \alpha_{\operatorname{sum}}$ and $\bPro_\bA(\bD=\bH)\leq \alpha_{j,i}$.
    Since $\varphi(\cdot,\cdot)$ increases as its arguments decrease, we have
    \begin{equation*}
        \bExp_\bH\left[ \bell_\bH(T) - \bell_\bA(T) \right] \geq \varphi(\alpha_{\operatorname{sum}},\alpha_{j,i}).
    \end{equation*}
    Combining the above, we have
    \begin{equation*}
        \bExp_\bH[T] \geq \frac{\varphi(\alpha_{\operatorname{sum}},\alpha_{j,i})}{\bI_\bH(\bA)}
    \end{equation*}
    for all $i\neq j\in[M]$ and $\bA\in\Alt_{i,j}(\bH,\cA)$. 
    Taking the worst case completes the proof.
\end{proof}

\begin{proof} [Proof of Theorem \ref{theorem, error control}]
    Fix arbitrary $\cA$, $\bH\in\cA$ and $\ba$.
    We first show almost sure finiteness.
    It is clear that 
    \begin{equation} \label{hat T(bH)}
    \begin{aligned}
        \hat T & \leq \hat T(\bH) := \inf_{n\geq 1} \Big\{
        \hat\bH(n) = \bH \text{ and } \bell_{\bH}(n) - \max_{\bA\in\widetilde\Alt_{i,j}(\bH,\cA)} \bell_\bA(n) \geq a_{j,i} \,\,\forall\,i\neq j\in[M]
        \Big\}.
    \end{aligned}
    \end{equation}
    Since, for any $i\neq j\in[M]$ and $\bA\in\Alt_{i,j}(\bH,\cA)$, the process 
    \begin{equation*}
        \bell_\bH(n)-\bell_\bA(n) = \sum_{k\in\bH\triangle\bA}(\ell^k_{H^k}(n)-\ell^k_{A^k}(n)),\,n\geq 1
    \end{equation*}
    has i.i.d. increments with positive mean $\bI_\bH(\bA)$, we have $\hat T\leq\hat T(\bH)<\infty$ almost surely under $\bPro_\bH$.

    We then show upper bounds on the error probabilities. 
    Fix arbitrary $i\neq j\in[M]$.
    It is clear that
    \begin{equation*}
    \begin{aligned}
        & \, \{H_i\cap \hat D_j\neq\emptyset\} = \{\hat\bD\in\Alt_{i,j}(\bH,\cA)\} = \bigcup_{\bA\in\Alt_{i,j}(\bH,\cA)}\{\hat\bD=\bA\}.
    \end{aligned}
    \end{equation*}
    Fix arbitrary such $\bA$.
    By definition of the procedure and the fact that $\bA\in\Alt_{i,j}(\bH,\cA)$ is equivalent to $\bH\in\Alt_{j,i}(\bA,\cA)$, we have
    \begin{equation} \label{last step in the proof of error control}
    \begin{aligned}
        & \, \{\hat\bD=\bA\} \subseteq 
        \Big\{\bell_{\bA}(\hat T) - \max_{\bA'\in\Alt_{j,i}(\bA,\cA)} \bell_{\bA'}(\hat T) \geq a_{i,j}\Big\} \\
        \subseteq & \, \Big\{\bell_{\bA}(\hat T) - \bell_{\bH}(\hat T) \geq a_{i,j}\Big\} 
        \subseteq \Big\{ \exists\, n\geq 1, \bell_{\bA}(n) - \bell_{\bH}(n) \geq a_{i,j}\Big\}.
    \end{aligned}
    \end{equation}
    One may check by definition that $\{\bell_\bA(n)-\bell_\bH(n):n\geq 1\}$ is a martingale under $\bPro_\bH$. 
    So the desired upper bound follows from Ville's inequality and the union bound. 
\end{proof}

\begin{proof} [Proof of Theorem \ref{theorem, AUB}]
    Fix arbitrary $\cA$ and $\bH\in\cA$.
    Also fix arbitrary $\ba$ and $\epsilon\in(0,1)$, and denote the right-hand-side of \eqref{expression, AUB} as $N_\bH(\ba,\cA)$.
    Recall the stopping time in \eqref{hat T(bH)}. 
    It is clear that 
    \begin{equation*}
    \begin{aligned}
        & \, \bExp_\bH[\hat T] \leq \bExp_\bH[\hat T(\bH)] = \sum_{n=0}^\infty \bPro_\bH(\hat T(\bH)>n) \\
        \leq & \, N_\bH(\ba,\cA)/(1-\epsilon) + \sum_{n=N_\bH(\ba,\cA)/(1-\epsilon)}^\infty \bPro_\bH(\hat T(\bH)>n).
    \end{aligned}        
    \end{equation*}
    For any $n\geq N_\bH(\ba,\cA)/(1-\epsilon)$, by definition we have
    \begin{equation*}
    \begin{aligned}
        & \{\hat T(\bH)>n\} \subseteq \bigcup_{k\in[K]} \bigcup_{j\in[M]\backslash\{H^k\}} \{ \ell^k_{H^k}(n) \leq \ell^k_j(n) \} \cup \bigcup_{i\neq j\in[M]} \bigcup_{\bA\in\widetilde\Alt_{i,j}(\bH,\cA)} \Big\{ \bell_\bH(n)-\bell_\bA(n) < a_{j,i} \Big\}.
    \end{aligned}
    \end{equation*}
    For any $k\in[K]$ and $j\in[M]\backslash\{H^k\}$, we have
    \begin{equation*}
        \{\ell^k_{H^k}(n) \leq \ell^k_j(n)\} \subseteq \Big\{ \frac{1}{n}\midbig( \ell^k_{H^k}(n) - \ell^k_j(n) \midbig) \leq 0 \Big\}.
    \end{equation*}
    For any $i\neq j\in[M]$ and $\bA\in\widetilde\Alt_{i,j}(\bH,\cA)$, since
    \begin{equation*}
        n\geq \frac{a_{j,i}}{\bI_{i,j}(\bH,\cA)}\frac{1}{1-\epsilon} \geq \frac{a_{j,i}}{\bI_\bH(\bA)} \frac{1}{1-\epsilon},
    \end{equation*}
    we have 
    \begin{align}
        & \, \Big\{ \bell_\bH(n)-\bell_\bA(n) < a_{j,i} \Big\} \label{the step in the proof of AUB} \\
        \subseteq & \, \Big\{ \frac{1}{n} \sum_{k\in\bH\triangle\bA} \midbig( \ell^k_{H^k}(n) - \ell^k_{A^k}(n) \midbig) \leq (1-\epsilon) \bI_\bH(\bA) \Big\} \nonumber \\
        \subseteq & \, \bigcup_{k\in\bH\triangle\bA} \Big\{ \frac{1}{n} \midbig( \ell^k_{H^k}(n) - \ell^k_{A^k}(n) \midbig) \leq (1-\epsilon) I^k_{H^k,A^k} \Big\}. \nonumber
    \end{align}
    So
    \begin{equation*}
    \begin{aligned}
        \{\hat T(\bH)>n\} 
        \subseteq \bigcup_{k\in[K]} \bigcup_{j\in[M]\backslash\{H^k\}} \Big\{ \frac{1}{n}\midbig( \ell^k_{H^k}(n) - \ell^k_{A^k}(n) \midbig) \leq (1-\epsilon) I^k_{H^k,j} \Big\}.
    \end{aligned}
    \end{equation*}
    By the union bound and Lemma \ref{lemma for AUB}, it follows that 
    \begin{equation*}
    \begin{aligned}
        \bExp_\bH[\hat T] & \leq N_\bH(\ba,\cA)/(1-\epsilon) + \sum_{n=1}^\infty \sum_{\substack{k\in[K] \\ j\in[M]\backslash\{H^k\} }} \bPro_\bH\Big( \frac{1}{n}\midbig( \ell^k_{H^k}(n) - \ell^k_{A^k}(n) \midbig) \leq (1-\epsilon) I^k_{H^k,j} \Big) \\
        & = N_\bH(\ba,\cA)/(1-\epsilon) + constant(\epsilon),
    \end{aligned}
    \end{equation*}
    where $constant(\epsilon)$ represents a constant irrelevant with $\ba$.
    First letting $a_\min\to\infty$ and then letting $\epsilon\to 0$ complete the proof. 
\end{proof}

\section{Proofs in Section \ref{sec: extension to composite hypotheses}} \label{appendix: proofs related to composite hyps}
\begin{proof} [Proof of Theorem \ref{theorem, composite, universal lower bound}]
    Fix $\cA$, $\btheta\in\bTheta_\cA$ and $\balpha$. 
    It is clear that, for any $\{v^k_j\in\Theta^k_j: j\in[M]\backslash\{H^k(\btheta)\},\,k\in[K]\}$, any procedure that solves the composite hypothesis testing problem considered in this section also solves with the same or better reliability the simple hypothesis testing problem in the previous sections with hypotheses
    \begin{equation*}
        f^k=f^k_{\theta^k} \text{ and } f^k=f^k_{v^k_j} \text{ for } j\in[M]\backslash\{H^k(\btheta)\}
    \end{equation*}    
    for every stream $k\in[K]$.
    Thus, applying Theorem \ref{theorem, LB}, we have
    \begin{equation*}
    \begin{aligned}
         \cL_\btheta(\balpha,\cA) \geq \sup_{\substack{v^k_j\in\Theta^k_j \text{ for } k\in[K] \\ \text{and } j\in[M]\backslash\{H^k(\btheta)\} }} \max_{i\neq j\in[M]} \frac{\varphi(\alpha_{\operatorname{sum}}, \alpha_{j,i})}{\min\limits_{\bA\in\widetilde\Alt_{i,j}(\bH(\btheta),\cA)} \sum\limits_{k\in\bH(\btheta)\triangle\bA} I^k_{\theta^k,v^k_{A^k}} },
    \end{aligned}
    \end{equation*}
    which simplifies to the desired form.
\end{proof}

\begin{proof} [Proof of Theorem \ref{theorem, composite, error control}]
    Since condition \eqref{composite, complete assumption} implies that the evidence almost surely increases at least linearly in $n$ whereas condition \eqref{composite, rate of beta} implies that the thresholds increase sub-linearly in $n$ for any $\balpha$, the almost sure finiteness follows.
    The proof of error control follows the same steps as the proof of Theorem \ref{theorem, error control}, with condition \eqref{composite, ville-type assumption} applied to the last step of \eqref{last step in the proof of error control}.
\end{proof}

\begin{proof} [Proof of Theorem \ref{theorem, composite, AUB}]
    Based on condition \eqref{composite, rate of beta}, write $\beta(n,\alpha)$ as $|\log\alpha|(1+c_1(\alpha))+c_2(n)n$, where $c_1(\alpha)\to 0$ as $\alpha\to 0$ and $c_2(n)\to 0$ as $n\to\infty$. 
    The proof of asymptotic upper bound follows the same steps as that of Theorem \ref{theorem, AUB}, except that in step \eqref{the step in the proof of AUB}, we have
    \begin{equation*}
    \begin{aligned}
        & \, \Big\{ \bell_{\bH(\btheta)}(n) - \bell_\bA(n) < \beta\midbig(n,\alpha_{j,i}/b_{j,i}(\cA)\midbig) \Big\} \\
        \subseteq & \, \Big\{ \frac{1}{n} \sum_{k\in\bH(\btheta)\triangle\bA} \midbig( \ell^k_{H^k(\btheta)}(n) - \ell^k_{A^k}(n) \midbig) \leq  (1+c_1(\alpha_{j,i}))(1-\epsilon)\bI_\btheta(\bA) + c_2(n) \Big\} \\
        \subseteq & \, \Big\{ \frac{1}{n} \sum_{k\in\bH(\btheta)\triangle\bA} \midbig( \ell^k_{H^k(\btheta)}(n) - \ell^k_{A^k}(n) \midbig) \leq (1-\epsilon/2)\bI_\btheta(\bA) \Big\}
    \end{aligned}
    \end{equation*}
    when $\alpha_\max$ is small enough and, thus, $n$ is large enough.
\end{proof}

\section{Supporting lemmas}
\begin{lemma} \label{lemma for AUB}
    Under condition \eqref{only assumption}, for any $\epsilon\in(0,1)$,
    \begin{equation*}
        \Pro^k_i\Big( \frac{1}{n}\midbig( \ell^k_i(n) - \ell^k_j(n) \midbig) \leq (1-\epsilon) I^k_{i,j} \Big) < \infty.
    \end{equation*}
\end{lemma}
\begin{proof}
    We suppress the lower and upper indices and denote the moment generating function of $\ell^k_i(1) - \ell^k_j(1)$ as 
    \begin{equation*}
        M(\theta) := \Exp^k_i\left[ \exp\{\theta(\ell^k_i(1) - \ell^k_j(1))\} \right] \text{ for } \theta\in\mathbb{R}.
    \end{equation*}
    It is clear that $M(0)=M(-1)=1$ which, by the property of moment generating functions (see, e.g., \cite{Dembo_Zeitouni_LDPBook}), implies that $M(\theta)$ is finite and continuous in $[-1,0]$, 
    and that $M'(0-)=I^k_{i,j}>0$.
    For any $x<I^k_{i,j}$, consider function 
    \begin{equation*}
        g_x(\theta) := \theta x - \log M(\theta) \text{ for } \theta\in[-1,0].
    \end{equation*}
    It is clear that $g_x(0)=0$ and $g'_x(0-)=x-I^k_{i,j}<0$, so there must exist a $\theta\in(-1,0)$ so that $g_x(\theta)>0$. 
    That is, we have
    \begin{equation*}
        \sup_{\theta\in(-1,0)} \midbig\{ \theta x-\log M(\theta) \midbig\} >0 \text{ for all } x<I^k_{i,j}.
    \end{equation*}

    Fix $\epsilon\in(0,1)$.
    By Chernoff's bound, for any $\theta\in(-1,0)$,
    \begin{equation*}
    \begin{aligned}
        & \, \Pro^k_i\Big( \frac{1}{n}\midbig( \ell^k_i(n) - \ell^k_j(n) \midbig) \leq (1-\epsilon) I^k_{i,j} \Big) \\
        = & \, \Pro^k_i\Big( \exp\midbig\{ \theta\midbig(\ell^k_i(n) - \ell^k_j(n)\midbig) \midbig\} \geq \exp\midbig\{ n\theta(1-\epsilon)I^k_{i,j} \midbig\} \Big) \\
        \leq & \, \Exp^k_i\Big[ \exp\midbig\{ \theta\midbig(\ell^k_i(n) - \ell^k_j(n)\midbig) \midbig\} \Big] \exp\midbig\{ -n\theta(1-\epsilon)I^k_{i,j} \midbig\} \\
        = & \, M(\theta)^n \exp\midbig\{ -n\theta(1-\epsilon)I^k_{i,j} \midbig\} \\
        = & \, \exp\midbig\{ -n\midbig( \theta(1-\epsilon)I^k_{i,j} - \log M(\theta) \midbig) \midbig\}.
    \end{aligned}
    \end{equation*}
    A sufficient condition for this being summable is 
    \begin{equation*}
        \sup_{\theta\in(-1,0)} \midbig\{ \theta(1-\epsilon)I^k_{i,j} - \log M(\theta) \midbig\} > 0,
    \end{equation*}
    which has been proved to be true.
\end{proof}

\bibliographystyle{chicago}
\bibliography{main}

\end{document}